\documentclass{sig-alternate}
\usepackage{subfig}
\usepackage{setspace}
\allowdisplaybreaks
\usepackage{float}
\usepackage{epstopdf}

\begin{document}
%
\conferenceinfo{EMSOFT'14}{, October 12-17, 2014, New Delhi, India}

\title{Supporting Read/Write Applications in Embedded Real-time Systems via Suspension-aware Analysis\titlenote{}}
%
%
%
%
%

\numberofauthors{2} 
%
\author{
%
%
\alignauthor
}

\maketitle
\begin{abstract}
In many embedded real-time systems, applications often interact with I/O devices via read/write operations, which may incur considerable suspension delays. Unfortunately, prior analysis methods for validating timing correctness in embedded systems become quite pessimistic when suspension delays are present. In this paper, we consider the problem of supporting two common types of I/O applications in a multiprocessor system, that is, write-only applications and read-write applications. For the write-only application model, we present a much improved analysis technique that results in only $O(m)$ suspension-related utilization loss, where $m$ is the number of processors. For the second application model, we present a flexible I/O placement strategy and a corresponding new scheduling algorithm, which can completely circumvent the negative impact due to read- and write-induced suspension delays. We illustrate the feasibility of the proposed I/O-placement-based schedule via a case study implementation. Furthermore, experiments presented herein show that the improvement with respect to system utilization over prior methods is often significant.
\end{abstract}




\section{Introduction}
\label{sec:intro}
Applications that incur read and/or write operations are commonly seen in embedded real-time systems. A typical data processing application may need to write data to the disk after performing computation on CPU. Such read and write operations cause non-negligible suspension delays during an application's execution, i.e., an application is suspended by the operating system while waiting for the completion of the I/O operation. For example, delays introduced by disk I/O range from 15$\mu s$ (for NAND flash) to 15$ms$ (for magnetic disks) per read~\cite{kang:disktime}.

Unfortunately, such delays cause intractability in validating applications' timing correctness, even in uniprocessor systems~\cite{ridouard:unihard}. If applications require hard real-time (HRT) constraints (i.e., meeting deadlines, which is an underlying requirement in many embedded real-time systems), then, in the worst-case, significant utilization of processors have to be sacrificed in order to provide such timing guarantee. Consider an example task system with two identical recurrent tasks $\tau_1$ and $\tau_2$ running on a uniprocessor platform. Each released job in $\tau_1$ and $\tau_2$ first spends $5 ms$ in reading data from the disk, then spends $5 ms$ in performing computation, and finally spends $5 ms$ in writing data to the disk. The relative deadline of these two tasks is set to be $15 ms$. From the earliest-deadline-first (EDF) schedule shown in Fig. \ref{fig:simpleexample}, $\tau_2$ misses its deadline while the total utilization of the system is low (i.e., each task only requires 5/15 of the processor capacity because suspensions do not occupy CPU). In this paper, we consider the problem of scheduling and analyzing HRT applications that contain I/O operations in a multiprocessor embedded real-time system. We specifically focus on two common types of such applications, i.e., write-only applications that incur only write operations, and read-write applications that incur both read and write operations.

\begin{figure}[t]
	\begin{center}
	\includegraphics[width=3in]{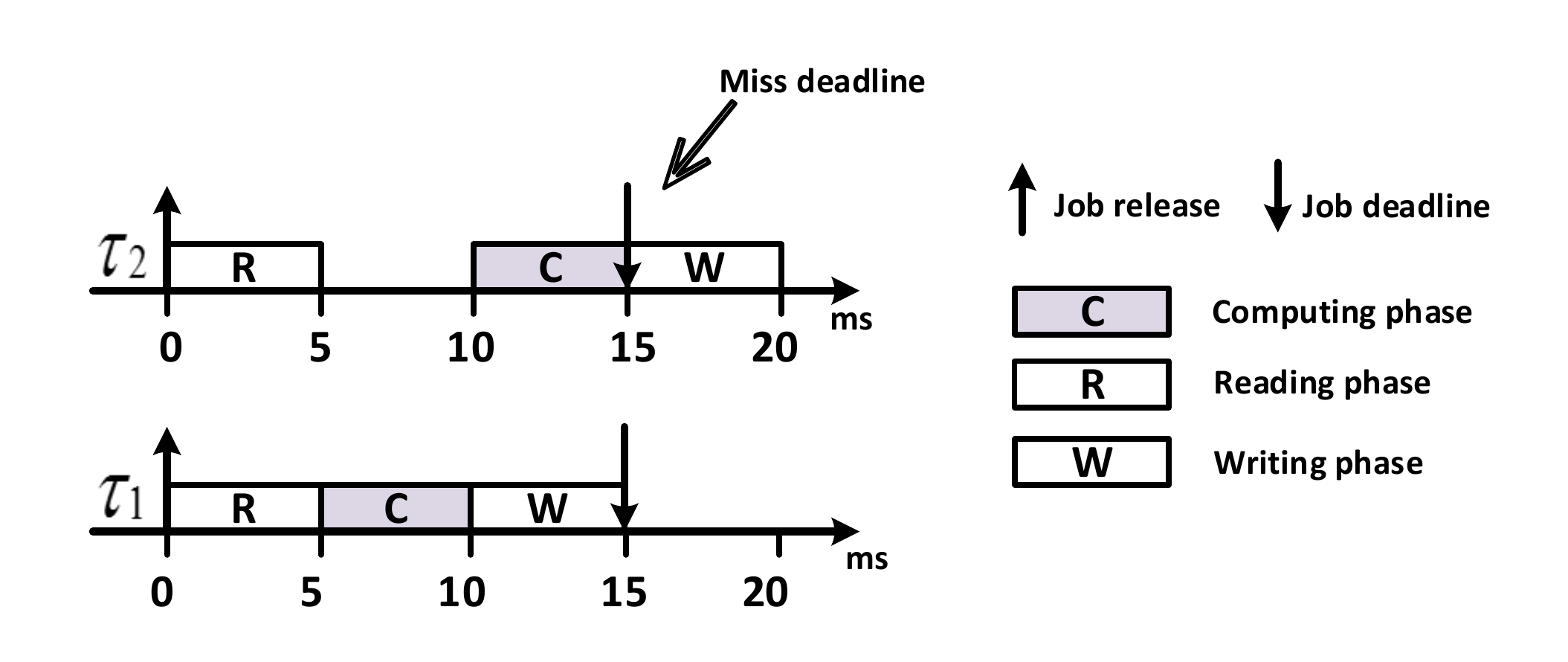} 
	\end{center} 
\vspace{-5mm}
\caption{Example task model}
\vspace{-5mm}
\label{fig:simpleexample}
\end{figure}

To deal with the read- and/or write-induced suspensions, perhaps the most commonly used approach is  \textit{suspension-oblivious} analysis, which simply treats suspension as computation by integrating suspension length into per-task worst-case computation time requirements. However, this approach yields $\Omega(n)$ suspension-related utilization loss where $n$ is the number of tasks that may suspend in the system. Significant system utilization may be sacrificed in the worst-case under this approach if the number of suspending tasks is large or suspension delays are long. The alternative is to explicitly consider suspensions in the scheduling analysis; this is known as \textit{suspension-aware} analysis. Previous research~\cite{cong:O(m)} has demonstrated the advantage of using suspension-aware analysis over suspension-oblivious analysis in many scenarios.

We thus consider in this paper designing new suspension-aware analysis techniques to improve system utilization. We focus on \textit{global-scheduling} approaches where tasks may migrate among processors (as opposed to \textit{partitioned-scheduling} where tasks are statically assigned to processors). Specifically, we study the global earliest-deadline-first(GEDF) \allowbreak scheduling algorithm herein, but our proposed techniques can also be extended to other fixed-job-priority global scheduling algorithms. We first present an improved analysis technique for write-only applications. For read-write applications, our observation is that if the time at which applications' read and write operations occur is not controllable, then utilization loss is fundamental. For example, as seen in Fig. \ref{fig:simpleexample}, regardless of how we prioritize the two tasks, one of them inevitably misses the deadline. Motivated by this observation, we design a flexible I/O placement policy, which allows the scheduler to judiciously control the time at which read and write operations occur. In this way, the negative impact due to read- and write-induced suspensions can be alleviated.

\paragraph*{\textbf{Overview}}For the soft real-time(SRT) case (i.e., only requiring bounded response times), an overview of the work in scheduling task systems with suspensions on multiprocessors can be found in \cite{cong:rtss09,liu2010improving,cong:O(m)}. But such technique cannot be applied to the analysis of the HRT case. For the HRT case, several works has been focused on periodic tasks that may suspend at most once on a uniprocessor \cite{lakshmanan2010scheduling,palencia1998schedulability,palencia2005response}. On multiprocessors, \cite{cong:hrt_sus} presents the only existing global suspension-aware analysis for sporadic HRT suspending task systems scheduled under global fixed-priority schedulers. However, the resulting schedulability tests require  pseudo-polynomial time complexity and may be pessimistic in many scenarios.

\paragraph*{\textbf{Contributions}}The existing suspension-oblivious and sus-pension-aware approaches for supporting applications containing read and write operations are pessimistic. In order to support such applications in a more efficient way, we present in this paper new suspension-aware analysis techniques for two common application models. For write-only applications, our proposed analysis techniques results in only $O(m)$ suspension-related utilization loss. To the best of our knowledge, this is the first analysis technique with a provable $O(m)$ suspension-related utilization loss for HRT suspending task systems. For applications with both read and write operations, we design a controllable I/O placement policy and a corresponding global EDF-based scheduling algorithm. We prove that the proposed I/O-placement-based scheduling technique can completely circumvent the negative impact due to read- and write-induced suspension. The feasibility of implementing this I/O-placement-based schedule in practice is demonstrated via a case study. As demonstrated by experiments, our proposed techniques significantly improve upon prior methods with respect to system utilization.

\paragraph*{\textbf{Organization}}The rest of the paper is organized as follows. We define the system model in Sec. \ref{sec:sysmodel}. In Sec. \ref{sec:wo}, we present our analysis for write-only applications. In Sec. \ref{sec:rw}, we present our I/O placement policy for read-write applications and the corresponding scheduling algorithm and schedulability test. In Sec. \ref{sec:case}, we provide a case study to show the improvement with respect to reducing response time and the feasibility of our I/O- placement-based schedule. In Sec. \ref{sec:woexp}, we experimentally evaluate the proposed schedulability test. Sec. \ref{sec:con} concludes.

\vspace{3mm}

\begin{figure}[t]
	\begin{center}
	\includegraphics[width=3in]{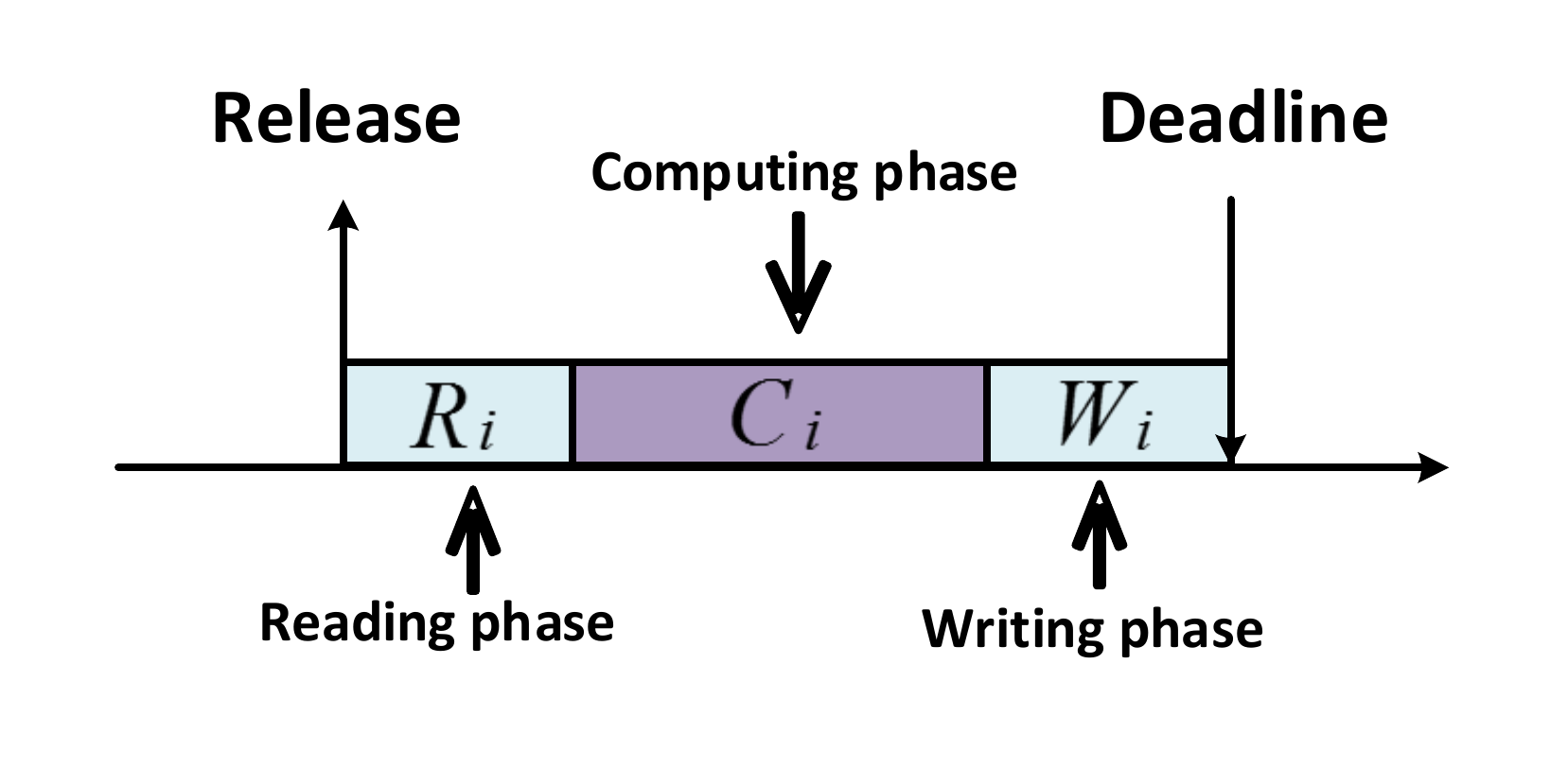} 
	\end{center} 
\vspace{-5mm}
\caption{Example read-write task pattern}
\vspace{-5mm}
\label{fig:rwpattern}
\end{figure}

\section{System Model}
\label{sec:sysmodel}

In this section, we formally define the system model. We first present the general task models for applications with read/write operations. Then we specifically define task models for read-write applications and write-only applications.

\paragraph*{\textbf{General task model}} An embedded real-time systems can be represented as a number of sporadic tasks that are invoked recurrently. Each invocation of a task is called a job and there is a minimum time between two consecutive job releases of a task. For each task, the time taken by read/write operations is present as suspensions on processors. Thus, an application with read/write operations is generally modeled as a suspending sporadic task. Let $\tau=\lbrace \tau_1,...,\tau_n \rbrace$ denote the task system that contains the set of $n$ independent suspending sporadic tasks. Let $T_i$ be the period of task $\tau_i$ where $T_i$ is the minimum time between two consecutive job releases of $\tau_i$. For each task $\tau_i$, let $C_i$ and $S_i$ denote the worst-case computation time and worst-case suspension time, respectively. Define the utilization $U_i$ of task $\tau_i$ as the ratio of computation time $C_i$ to its period $T_i$ ( i.e., $U_i=\frac{C_i}{T_i}$) and the suspension ratio $V_i$ of task $\tau_i$ as the ratio of suspension time $S_i$ to its period $T_i$ ( i.e., $V_i=\frac{S_i}{T_i}$). We require 
\begin{equation}
\label{eq:req1_2}
U_i+V_i\leq 1, 
\end{equation}
for otherwise, task $\tau_i$ must miss its deadline in the worst case. Let $U_{sum}= \sum_{i=1}^{n}U_{i}$ and $V_{sum}= \sum_{i=1}^{n}V_i$, where $U_{sum}$ is the system utilization.

Let $\tau_{i,j}$ be the $j$th job released by task $\tau_i$; let $r_{i,j}$ and $d_{i,j}$ be the corresponding release time and deadline. We consider the implicit-deadline task systems where $d_{i,j}-r_{i,j}=T_i$. 

Different kinds of applications may have varied operation interleaving patterns. Next, we specifically define the task models for read-write applications and write-only applications according to their operation interleaving patterns.

\paragraph*{\textbf{Read-write task model}}The most common operation interleaving pattern is to first read data from I/O devices, then perform computation based upon the data, and finally write the result back to I/O devices. As shown in Fig. \ref{fig:rwpattern}, each read-write task has three phases, a reading phase, a computing phase, and a writing phase where the reading and writing phases are modeled as suspensions. For each read-write task $\tau_i$, let $R_i$, $C_i$ and $W_i$ denote the total length of its reading phase, computing phase, and writing phase, respectively. Then each read-write task $\tau_i$ could be represented as $\tau_i=(R_i,C_i,W_i,T_i)$. For a read-write task $\tau_i$, we have $U_i= \frac{C_i}{T_i}$ and $V_i= \frac{R_i+W_i}{T_i}$.

\paragraph*{\textbf{Write-only task model}} The write-only task model is used to represent write-only applications. As shown in Fig. \ref{fig:wopattern}, each write-only task has three phases, where the first phase and the last phase are computing phases, and the second phase is a writing phase. For each write-only task $\tau_i$, let $C_{i,1}$, $W_1$ and $C_{i,2}$ denote the length of its first computing phase, writing phase, and the second computing phase, respectively. Each write-only task $\tau_i$ can thus be represented as $\tau_i=(C_{i,1}, W_i, C_{i,2}, T_{i})$. Similarly, for a write-only task $\tau_i$, we have $U_i=\frac{C_{i,1}+C_{i,2}}{T_i}$ and $V_i= \frac{W_i}{T_i}$. Let $\delta_i = \frac{W_i}{C_{i,1}}$ denote the ratio of the length of the writing phase to the length of the first computing phase. We will use $\delta_i$ later in the analysis given in Sec. \ref{subsec:rwanalysis}.

\paragraph*{\textbf{System model}}We assume that the platform is comprised of $m$ identical processors. We consider discrete time system. Interval with unit length is called unit interval.
\newdef{definition}{\textbf{Definition}}
\begin{definition}
\label{def:busy}
A unit interval $[t,t+1)$ is \textit{busy} (resp. \textit{non-busy}) for a job set $J$ if all $m$ processors execute jobs in $J$ during $[t,t+1)$. A time interval $[a,b)$ is \textit{busy} (resp. \textit{non-busy}) for a job set $J$ if each (resp. not all) unit interval within $[a,b)$ is busy for $J$. For conciseness, if we say an interval is busy without referring any job set, we mean it is busy for the set of all jobs in the task system. 
\end{definition}

\begin{definition}
\label{def:pending}
If job $\tau_{i,j}$ has been released but has not finished its last phase at time instant $t$, we say $\tau_{i,j}$ is \textit{pending} at $t$. If job $\tau_{i,j}$ has been released but has not finished its all computation at time instant $t$, we say $\tau_{i,j}$ is \textit{comp-pending} at $t$. If job $\tau_{i,j}$ has been released but has not finished its all suspension at time instant $t$, we say $\tau_{i,j}$ is \textit{sus-pending} at $t$. 
\end{definition}

\begin{definition}
\label{def:compavailable}
At time instant $t$, if job $\tau_{i,j}$ is comp-pending and it is able to perform computation at $t$, we say $\tau_{i,j}$ is \textit{comp-available} at $t$. Note that if a job is suspended by read/write operations, it cannot perform computation even if it is comp-pending. 
\end{definition}

\begin{definition}
\label{def:prempted}
For a unit interval $[t,t+1)$, if a job $\tau_{i,j}$ is comp-available at $t$ but does not compute in $[t,t+1)$, we say $\tau_{i,j}$ is comp-preempted in $[t,t+1)$.
\end{definition}

\paragraph*{\textbf{Scheduling algorithm}}In this paper, we focus on global earliest-deadline-first(GEDF) scheduling algorithm defined as follows.

\begin{definition}
\label{def:GEDF}
At each time instant, GEDF selects $m$ comp-available jobs with the earliest deadlines for computation. Ties are broken by index where tasks with lower indexes are favored. Jobs are allowed to migrate among different processors.
\end{definition}

\begin{figure}[t]
	\begin{center}
	\includegraphics[width=3in]{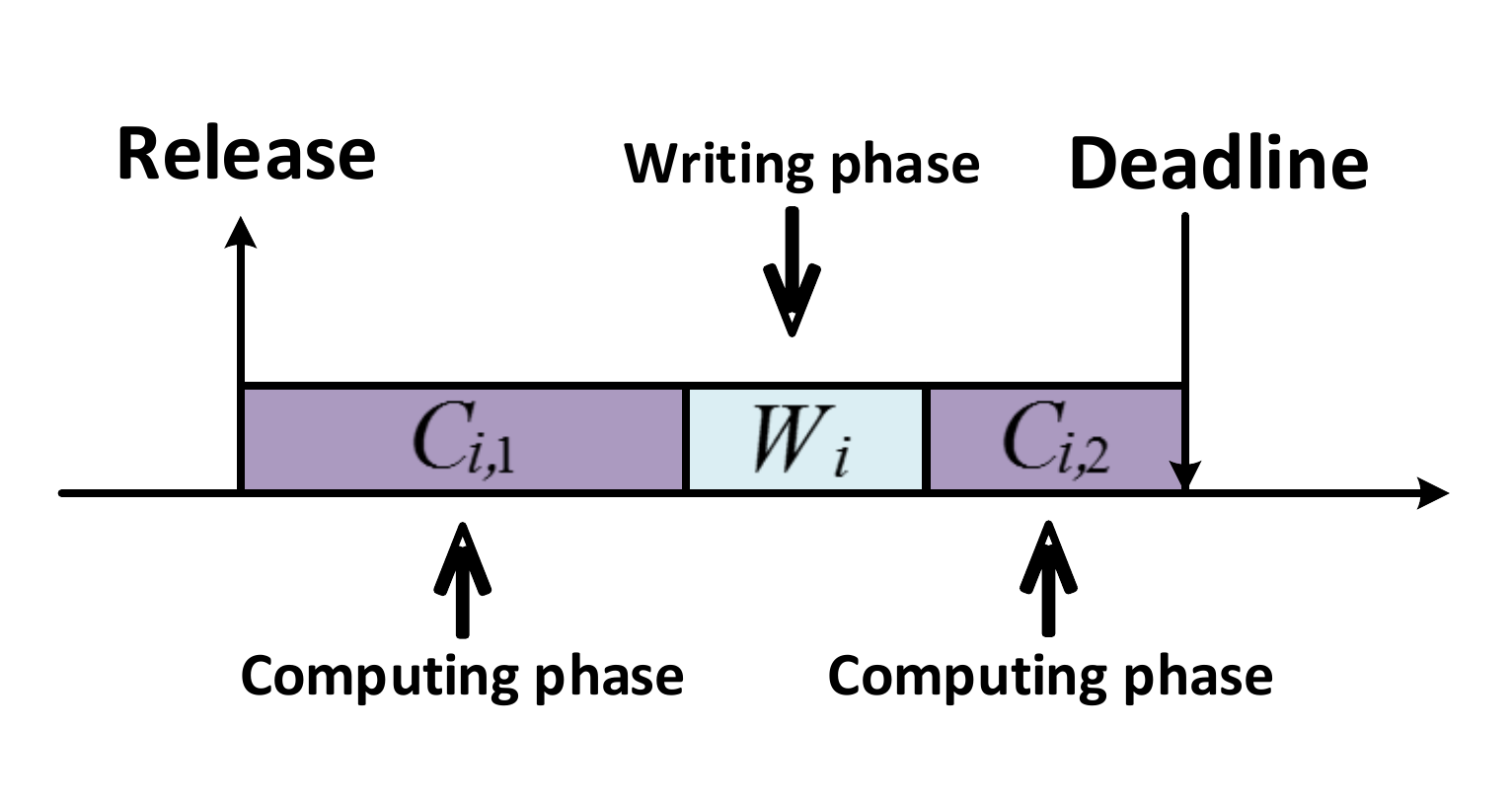} 
	\end{center} 
\vspace{-5mm}
\caption{Example write-only task pattern}
\vspace{-5mm}
\label{fig:wopattern}
\end{figure}

\section{Supporting Write-only Applications}
\label{sec:wo}

In this section, we investigate the GEDF schedule for HRT write-only task systems. Our analysis draws the inspiration from the lag-based analysis technique presented in the seminal work of Devi \cite{devi2006soft}. Lag-based technique was originally designed to handle SRT task systems. It has been extensively applied to analyze different SRT suspending task systems in prior work \cite{cong:O(m)}. Based upon this technique, we develop a new suspension-aware analysis for HRT suspending task systems, which is the first of its kind to the best of our knowledge. 

We will first introduce the lag-based analysis technique, and then present our suspension-aware analysis and the resulting schedulability test. 

\subsection{Lag-based Analysis Technique}
\label{subsec:lag}

For any given write-only task system $\tau$, a \textit{processor share} (\textit{PS}) schedule is an ideal schedule for computation where each task $\tau_i$ performs computation with a speed equal to $U_i$ when it is comp-pending (which ensures that each of its jobs completes its computation exactly at its deadline). Note that suspensions are not considered in the PS schedule and a task could execute the second computing phase as long as it has finished the first computing phase. A valid PS schedule exists for $\tau$ if $U_{sum} \leq m$ holds. 

Fig. \ref{fig:PSschedule} illustrates the PS schedule of the task system in Fig. \ref{fig:simpleexample}. Each of the two tasks has a utilization equal to $1/3$ and thus shares $1/3$ of the processor capacity. Every job of the two tasks finishes its computation at its deadline and suspensions are not considered. We can see that the PS schedule is not a real schedule and only used to keep track of the computation for analysis purposes. 

\begin{figure}[th]
	\begin{center}
	\includegraphics[width=3in]{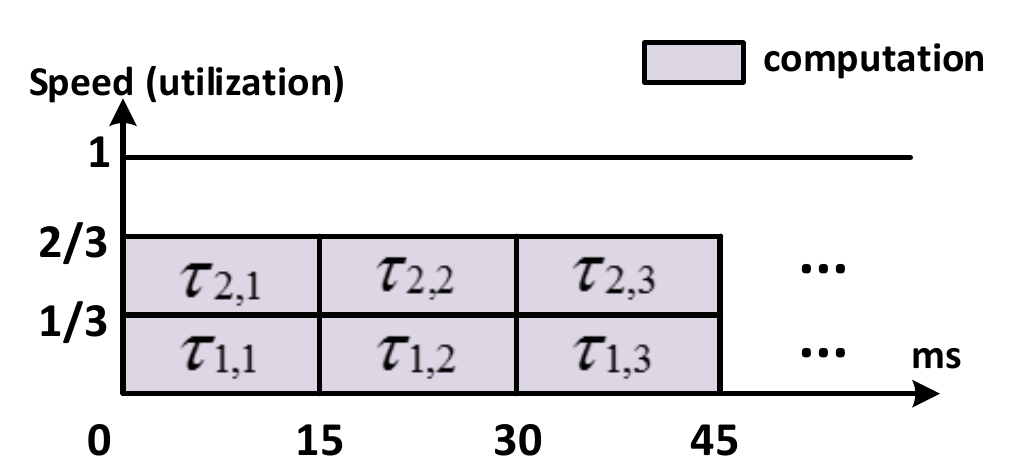} 
	\end{center} 
\vspace{-5mm}
\caption{Example PS schedule}
\vspace{-3mm}
\label{fig:PSschedule}
\end{figure}

In the lag-based analysis, schedulability tests are obtained by comparing the computation performed by tasks in the GEDF schedule $S$ and the PS schedule \textit{PS}. Let $A(\tau_{i,j}, t_1, t_2, S)$ denote the total computation performed by job $\tau_{i,j}$ under GEDF in $[t_1, t_2)$. Then, the total computation performed by each task $\tau_i$ and all tasks in $\tau$ in $[t_1, t_2)$ under GEDF is given by
\begin{eqnarray*}
A(\tau_i, t_1, t_2, S)=\sum_{j \geq 1} A(\tau_{i,j}, t_1, t_2, S)
\end{eqnarray*}
and
\begin{eqnarray*}
A(\tau, t_1, t_2, S)=\sum_{i = 1}^{n} A(\tau_i, t_1, t_2, S).
\end{eqnarray*}

$A(\tau_i, t_1, t_2, PS)$ and $A(\tau, t_1, t_2, PS)$ can be defined in a similar manner corresponding to $PS$ schedule. 

The difference between the computation performed by a job $\tau_{i,j}$ up to time $t$ in \textit{PS} and $S$, denoted \textit{the lag of job $\tau_{i,j}$ at time $t$}, is defined by
\begin{equation*}
lag(\tau_{i,j}, t, S) = A(\tau_{i,j}, 0, t, PS) - A(\tau_{i,j}, 0, t, S).
\end{equation*}
Similarly, the difference between the computation performed by a task $\tau_{i}$ up to time $t$ in  \textit{PS} and $S$, denoted \textit{the lag of task $\tau_{i}$ at time $t$}, is defined by
\begin{eqnarray}
\label{eq:lag for task}
lag(\tau_i, t, S) \hspace{-1.5mm} & = & \hspace{-2mm} \sum_{j \geq 1} lag(\tau_{i,j}, t, S) \nonumber \\
                \hspace{-1.5mm} & = & \hspace{-1.5mm} \sum_{j \geq 1} \big(A(\tau_{i,j}, 0, t, PS) - A(\tau_{i,j}, 0, t, S)\big)\hspace{-0.5mm}.
\end{eqnarray}
\normalsize

The \textit{LAG} for the task system $\tau$ at time $t$ is defined as
\begin{equation}
\label{eq:LAG for task set_1}
LAG(\tau, t, S) = \sum_{i=1}^{n} lag(\tau_i, t, S).
\end{equation}
Also, $LAG(\tau, t, S)$ can be represented as follow.
\begin{equation}
\label{eq:LAG for task set_2}
LAG(\tau, t, S) = A(\tau, 0, t, PS)- A(\tau, 0, t, S).
\end{equation}
\newtheorem{lemma}{\textbf{Lemma}}
\begin{lemma}
\label{lem:nonincrease}
If $[t_1,t_2)$ is a busy interval, then,
\begin{eqnarray*}
LAG(\tau, t_2, S) \leq LAG(\tau, t_1, S). 
\end{eqnarray*}
\end{lemma}

\begin{proof}
By Eq. (\ref{eq:LAG for task set_2}), 
\begin{eqnarray*}
&& LAG(\tau, t_2, S) - LAG(\tau, t_1, S)\\ \hspace{-1.5mm} 
& = & A(\tau, t_1, t_2, PS)- A(\tau, t_1, t_2, S)\\
& = & U_{sum} \cdot (t_2- t_1) - m \cdot (t_2-t_1)\\
& = & (U_{sum} - m) \cdot (t_2-t_1)\\
& \leq  & 0.
\end{eqnarray*}

Lemma \ref{lem:nonincrease} implies $LAG(\tau, t, S)$ cannot increase during a busy interval. 
\end{proof}

\subsection{Lag-based Analysis for HRT Write-only Task Systems}
\label{subsec:wo_analysis}

Now we analyze the schedulability for $n$ sporadic write-only tasks scheduled on $m$ processors under GEDF.

\begin{lemma}
\label{lem:lem2}
Consider job $\tau_{i,j}$ and a time instant $t > r_{i,j}$. Let $S_i^{*}$ denote the suspension of $\tau_{i,j}$ finished by $t$. Let $C_i^{*}$ denote the computation of $\tau_{i,j}$ performed by time instant $t$. Then, 
\begin{equation}
\label{eq:lem2}
\frac{S_i^{*}}{C_i^{*}} \leq \delta_i,                  
\end{equation}
where $\delta_i = W_i / C_{i,1}$, as defined in Sec. \ref{sec:sysmodel}.
\end{lemma}

\begin{proof}

\textbf{Case 1.} If $\tau_{i,j}$ has not finished its first computing phase at $t$, then $S_i^{*}=0$. Eq.  (\ref{eq:lem2}) clearly holds. 

\textbf{Case 2.}If $\tau_{i,j}$ has finished its first computing phase but has not finished its writing phase at $t$, then $S_i^{*} \leq  W_i$ and $C_i^{*}= C_{i,1}$. Thus,
\begin{equation*}
\frac{S_i^{*}}{C_i^{*}} =  \frac{S_i^{*}}{C_{i,1}} \leq \frac{W_i}{C_{i,1}} = \delta_i.                    
\end{equation*}

\textbf{Case 3.} If $\tau_{i,j}$ has finished its writing phase at $t$, then $C_i^{*} \geq  C_{i,1}$ and $S_i^{*} = W_i$. Thus,
 
\begin{equation*}
\frac{S_i^{*}}{C_i^{*}} =  \frac{W_i}{C_i^{*}} \leq \frac{W_i}{C_{i,1}} = \delta_i.                   
\end{equation*}

Lemma \ref{lem:lem2} proved.
\end{proof}

\begin{lemma}
\label{lem:proof1}
Let $L_i= (m-1) \cdot U_i + m \cdot U_i \cdot \delta_i$ and $L=max\{L_1,...,L_n\}$. If
\begin{equation}
\label{eq:as_sum}
U_{sum} \leq m-L,                   
\end{equation} 
and, for every $i$,
\begin{equation}
\label{eq:req1}
U_{i}\cdot (1+\delta_i)< 1.                      
\end{equation}

then no job misses its deadline under GEDF.
\end{lemma}
\begin{proof}
We prove this lemma by contradiction. Assume job $\tau_{i,j}$ is the first job that misses its deadline at $d_{i,j}$. If more than one job misses deadline at $d_{i,j}$, we choose the one with the highest priority. Since jobs with priorities lower than that of $\tau_{i,j}$ do not impact the scheduling of $\tau_{i,j}$, we get rid of such jobs from our task system.

Because $\tau_{i,j}$ has not finished its last phase at $d_{i,j}$, by the definition of $lag(\tau_i, d_{i,j}, S)$, $lag(\tau_i, d_{i,j}, S)> 0$. For every $k\neq i$, $lag(\tau_k, d_{i,j}, S) = 0$ because $\tau_{i,j}$ is the first job that misses its deadline. Thus,
\begin{eqnarray*}
LAG(\tau, d_{i,j}, S) &=& \sum_{k=1}^{n} lag(\tau_k, d_{i,j}, S) \\
&=& lag(\tau_i, d_{i,j}, S)\\
&>& 0.
\end{eqnarray*}

From time instant $0$, let $t^{*}$ be the earliest time instant such that
\begin{equation}
\label{eq:as_lag1}
LAG(\tau, t^{*}, S) > 0.                   
\end{equation}
Since $LAG(\tau, 0,S) = 0$ and $LAG(\tau, d_{i,j}, S) > 0$, $t^{*}$ is well defined and $d_{i,j} \geq t^{*} > 0$. By the definition of $LAG(\tau, t^{*}, S)$, there exists a task $\tau_k$ at $t^{*}$ such that $lag(\tau_k, t^{*}, S) >0$, which implies $\tau_k$ must have at least one pending job. Because  $\tau_{i,j}$ is the first job that misses deadline, $\tau_k$ has only one job pending at $t^{*}$. Let $\tau_{k,l}$ be this pending job of $\tau_k$ and $\tau_{k,l}$ is released at $r_{k,l}$. Because jobs of $\tau_k$ released before $\tau_{k,l}$ have finished all their phases, we have
\begin{eqnarray}
\label{eq:as_lag2}
&& lag(\tau_k, t^{*}, S) = lag(\tau_{k,l}, t^{*}, S)  > 0. 
\end{eqnarray}

There are three kinds of unit intervals to be considered in $[r_{k,l}, t^{*})$  as shown in Fig. \ref{fig:lemma3}: (1) $\tau_{k,l}$ suspends in it; (2) $\tau_{k,l}$ computes in it; (3) $\tau_{k,l}$ does not compute or suspend in it. Let $\beta_1$, $\beta_2$ and $\beta_3$ denote the set of each kind of unit intervals, respectively. Thus $\beta_1\cup \beta_2 \cup \beta_3= [r_{k,l}, t^{*})$ and they are pairwise disjoint. Let $B_1$, $B_2$ and $B_3$ denote the lengths of each set, respectively. Note that unit intervals in $\beta_3$ must be busy. Depending on the lengths of these sets, we have the following cases to consider.

\begin{figure}[th]
	\begin{center}
	\includegraphics[width=3.5in]{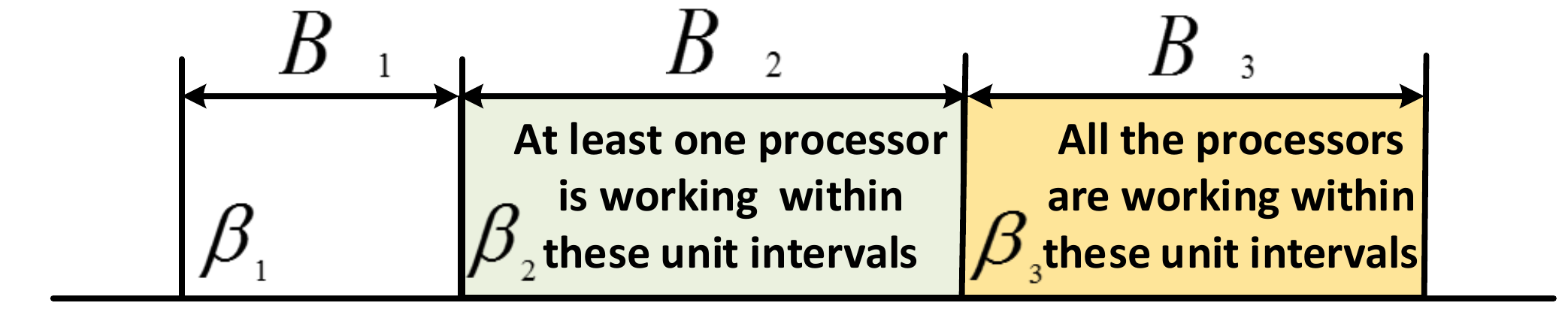} 
	\end{center} 
\vspace{-5mm}
\caption{Three sets of unit intervals in $[r_{k,l}, t^{*})$}
\vspace{-3mm}
\label{fig:lemma3}
\end{figure}

\paragraph*{\textbf{case A}} First, we discuss the cases when $B_1=0$, which implies $\tau_{k,l}$ does not suspend in $[r_{k,l}, t^{*})$.

\textbf{case A.1 : $B_1= B_2 = B_3= 0$.} In this case, $t^*= r_{k,l}$ and $lag(\tau_{k,l}, t^{*}, S)= 0$, which violates Eq.  (\ref{eq:as_lag2}).

\textbf{case A.2 : $B_1=0, B_2=0, B_3>0$.} In this case, $[r_{k,l}, t^{*})$ is a busy interval. By Lemma \ref{lem:nonincrease},
\begin{eqnarray*} 
LAG(\tau, t^{*}, S) \leq LAG(\tau, r_{k,l}, S) \leq 0, 
\end{eqnarray*}
which violates Eq.  (\ref{eq:as_lag1}). 

\textbf{case A.3 : $B_1=0, B_2>0, B_3=0$.} In this case, $\tau_{k,l}$ is executing during the time interval. So $lag(\tau_{k,l}, t^{*}, S) = U_k \cdot B_2-B_2 \leq 0$, which violates Eq. (\ref{eq:as_lag2}). 

\textbf{case A.4 : $B_1=0, B_2>0, B_3>0$.} In this case, first we consider $lag(\tau_{k,l}, t^{*}, S)$. By the definitions of $A(\tau_{k,l}, t_1, t_2, PS)$ and $A(\tau_{k,l}, t_1, t_2, S)$, we have
\begin{eqnarray}
\label{eq:c4_1_2}
A(\tau_{k,l}, r_{k,l}, t^{*}, PS) = (B_2+B_3) \cdot u_k
\end{eqnarray}
and $A(\tau_{k,l}, r_{k,l}, t^{*}, S)= B_2$ . By Eq. (\ref{eq:as_lag2}),
\begin{eqnarray*}
\label{eq:c4_2_2}
0 &<& lag(\tau_{k,l}, t^{*}, S)  \nonumber\\
&=& A(\tau_{k,l}, r_{k,l}, t^{*}, PS) - A(\tau_{k,l},  r_{k,l}, t^{*}, S)  \nonumber\\
&\leq& (B_2+B_3) \cdot U_k - B_2   \nonumber \\ 
&=& B_2 \cdot (U_k -1) + B_3 \cdot U_k,
\end{eqnarray*} 
which implies
\begin{equation}
\label{eq:c4_3_2}
B_2 < \frac{B_3 \cdot U_k}{1- U_k }.
\end{equation}
Now we consider $LAG(\tau, t^{*}, S)$. Because $t^{*}$ is the earliest time instant that $LAG(\tau, t^{*}, S)> 0$ and $r_{k,l} < t^{*}$, we have $LAG(\tau, r_{k,l}, S) \leq 0$. By the definition of LAG, 
\begin{eqnarray*}
& & LAG(\tau, t^{*}, S) \\
&=& LAG(\tau, r_{k,l}, S)+ A(\tau, r_{k,l}, t^{*}, PS) - A(\tau, r_{k,l}, t^{*}, S). 
\end{eqnarray*}
Thus, 
\begin{eqnarray}
\label{eq:c4_4_2}
LAG(\tau, t^*, S) \leq A(\tau, r_{k,l}, t^{*}, PS) - A(\tau, r_{k,l}, t^{*}, S).
\end{eqnarray}
Also we have
\begin{eqnarray}
\label{eq:c4_5_2}
A(\tau, r_{k,l}, t^{*}, PS) = U_{sum} \cdot (B_2+B_3)
\end{eqnarray} 
and because $\beta_3$ is a busy interval and $\tau_{k,l}$ is executing during $\beta_2$,
\begin{eqnarray}
\label{eq:c4_6_2}
A(\tau, r_{k,l}, t^{*}, S)  \geq m \cdot B_3+ B_2.
\end{eqnarray}
Therefore, by Eq. (\ref{eq:as_lag1}),
\begin{eqnarray*}
\label{eq:c4_7_2}
0 &<& LAG(\tau, t^{*}, S) \nonumber\\
&&{\lbrace \rm{by}~(\ref{eq:c4_4_2}) \rbrace} \nonumber\\ 
&\leq& A(\tau, r_{k,l}, t^{*}, PS) - A(\tau, r_{k,l}, t^{*}, S)  \nonumber\\
&&{\lbrace \rm{by}~(\ref{eq:c4_5_2}) \rbrace} \nonumber\\ 
&=& U_{sum} \cdot (B_2+B_3)- A(\tau, r_{k,l}, t^{*}, S) \nonumber \\ 
&&{\lbrace \rm{by}~(\ref{eq:c4_6_2}) \rbrace} \nonumber\\
&\leq& U_{sum} \cdot (B_2+B_3)- (m \cdot B_3+ B_2) \nonumber \\  
&=& (U_{sum}-1)\cdot B_2+ (U_{sum}-m) \cdot B_3\nonumber \\
&&{\lbrace \rm{by}~(\ref{eq:c4_3_2}) \rbrace} \nonumber\\ 
&<& \frac{(U_{sum}-1) \cdot B_3 \cdot U_k}{1- U_k }  + (U_{sum}-m) \cdot B_3
\end{eqnarray*}
By rearrangements, we have 
\begin{eqnarray*}
\label{eq:c4_8_2}
U_{sum}&>& m-(m-1)\cdot U_k \\
&>& m-(m-1)\cdot U_{max} \\
&>& m-L
\end{eqnarray*}
However, this violates Eq. (\ref{eq:as_sum}).

\paragraph*{\textbf{case B}} Secondly, we discuss the cases when $B_1>0$.

\textbf{case B.1 : $B_1>0 , B_2 = 0, B_3 \geq 0$.} In this case, $\tau_{k,l}$ writes data before doing computation, which violates the phase interleaving pattern of write-only tasks according to the write-only task model.

\textbf{case B.2 : $B_1 > 0 , B_2 > 0, B_3 = 0$.} In this case, 
\begin{eqnarray*}
&& A(\tau_{k,l}, r_{k,l}, t^{*}, PS)  \\
&=& (B_1+B_2) \cdot U_k  \\
&&{\lbrace \rm{by}~Lemma \ref{lem:lem2} \rbrace}\\ 
&{\leq}& (B_2 \cdot \delta_k+B_2) \cdot U_k \\
&=& B_2 \cdot (\delta_k+1) \cdot U_k  \\
&&{\lbrace \rm{by}~(\ref{eq:req1}) \rbrace} \\
&{<}& B_2, 
\end{eqnarray*}
and $A(\tau_{k,l}, r_{k,l}, t^{*}, S)= B_2$ . Thus 
$lag(\tau_{k,l}, t^{*}, S)=$\\ $A(\tau_{k,l}, r_{k,l}, t^{*}, PS) - A(\tau_{k,l},  r_{k,l}, t^{*}, S) < 0$, which violates Eq. (\ref{eq:as_lag2}).

\textbf{case B.3 : $B_1>0 , B_2 > 0, B_3 > 0$.} First we consider $lag(\tau_{k,l}, t^{*}, S)$. By the definitions of $A(\tau_{i,j}, t_1, t_2, PS)$ and $A(\tau_{i,j}, t_1, t_2, S)$, we have
\begin{eqnarray}
\label{eq:c6_1}
&& A(\tau_{k,l}, r_{k,l}, t^{*}, PS)  \nonumber \\ 
&=& (B_1+B_2+B_3) \cdot U_k  \nonumber\\
&&{\lbrace \rm{by}~Lemma \ref{lem:lem2} \rbrace} \nonumber\\ 
&{\leq}& (B_2 \cdot \delta_k+B_2 +B_3) \cdot U_k  \nonumber\\
&=& B_2 \cdot (\delta_k+1) \cdot U_k + B_3 \cdot U_k,
\end{eqnarray}
and $A(\tau_{k,l}, r_{k,l}, t, S)= B_2$ . By Eqs.  (\ref{eq:as_lag2}) and (\ref{eq:c6_1}),
\begin{eqnarray}
\label{eq:c6_2}
0 &<& lag(\tau_{k,l}, t^{*}, S)  \nonumber\\
&=& A(\tau_{k,l}, r_{k,l}, t^{*}, PS) - A(\tau_{k,l},  r_{k,l}, t^{*}, S)  \nonumber\\
&\leq& B_2 \cdot (\delta_k+1) \cdot U_k + B_3 \cdot U_k - B_2   \nonumber \\ 
&=& B_2 \cdot \big((\delta_k+1) \cdot U_k -1\big) + B_3 \cdot U_k 
\end{eqnarray} 
By Eqs. (\ref{eq:req1}) and (\ref{eq:c6_2}),
\begin{equation}
\label{eq:c6_3}
B_2 < \frac{B_3 \cdot U_k}{ 1-U_k \cdot (\delta_k+1)}.
\end{equation}
Now let us consider $LAG(\tau, t^{*}, S)$. Because $t^*$ is the earliest time instant that $LAG(\tau, t^{*}, S)> 0$, and $r_{k,l} < t^*$ by the definition of $\tau_{k,l}$, we thus have $LAG(\tau, r_{k,l}, S) \leq 0$. By Eq. (\ref{eq:LAG for task set_2}),  $LAG(\tau, t^{*}, S) = LAG(\tau, r_{k,l}, S)+ A(\tau, r_{k,l}, t^{*}, PS) - A(\tau, r_{k,l}, t^{*}, S)$. Thus, 

\begin{eqnarray}
\label{eq:c6_4}
LAG(\tau, t^{*}, S) \leq A(\tau, r_{k,l}, t^{*}, PS) - A(\tau, r_{k,l}, t^{*}, S).
\end{eqnarray}
By the definitions of $A(\tau, t_1, t_2, PS)$ and $A(\tau, t_1, t_2, S)$, we have
\begin{eqnarray}
\label{eq:c6_5}
A(\tau, r_{k,l}, t^{*}, PS) = U_{sum} \cdot (B_1+B_2+B_3)
\end{eqnarray} 
and
\begin{eqnarray}
\label{eq:c6_6}
A(\tau, r_{k,l}, t^{*}, S)  \geq m \cdot B_3+ B_2.
\end{eqnarray}
Thus, by Eq. (\ref{eq:as_lag1}), we have
\begin{eqnarray}
\label{eq:c6_7}
0 &<& LAG(\tau, t, S) \nonumber\\
&&{\lbrace \rm{by}~(\ref{eq:c6_4}) \rbrace} \nonumber\\ 
&\leq& A(\tau, r_{k,l}, t^{*}, PS) - A(\tau, r_{k,l}, t^{*}, S)  \nonumber\\
&&{\lbrace \rm{by}~(\ref{eq:c6_5}) \rbrace} \nonumber\\ 
&=& U_{sum} \cdot (B_1+B_2+B_3)- A(\tau, r_{k,l}, t^{*}, S) \nonumber \\ 
&&{\lbrace \rm{by}~(\ref{eq:c6_6}) \rbrace} \nonumber\\ 
&\leq& U_{sum} \cdot (B_1+B_2+B_3)- (m \cdot B_3+ B_2) \nonumber \\
&&{\lbrace \rm{by}~Lemma~\ref{lem:lem2} \rbrace} \nonumber\\   
&\leq& U_{sum}\cdot (B_2 \cdot \delta_k+B_2+B_3)- m \cdot B_3-B_2 \nonumber \\
&=& \big(U_{sum}\cdot (\delta_k+ 1)-1\big)\cdot B_2+ (U_{sum}-m) \cdot B_3\nonumber \\
&&{\lbrace \rm{by}~(\ref{eq:c6_3}) \rbrace} \nonumber\\ 
&<& \frac{\big(U_{sum}\cdot (\delta_k+1)-1\big) \cdot B_3 \cdot U_k}{1-(\delta_k+1) \cdot U_k }  + (U_{sum}-m) \cdot B_3.\nonumber
\end{eqnarray}
By rearrangements, we have 
\begin{eqnarray}
\label{eq:c6_8}
&&U_{sum} 
\nonumber \\
&>& m-\big((m-1)\cdot U_k + m \cdot U_k \cdot \delta_k\big) 
\nonumber \\
&>& m-L,
\end{eqnarray}
which violates Eq. (\ref{eq:as_sum}).

Thus far, we have discussed all of the possible cases and each case implies a contradiction. Lemma \ref{lem:proof1} thus follows. 
\end{proof}
Lemma \ref{lem:proof1} implies the following schedulability test.
\newtheorem{theorem}{\textbf{Theorem}}
\begin{theorem}
\label{theorem:wotest}
Any write-only task system $\tau$ can be successfully scheduled under GEDF on $m$  identical processors, provided $U_{i}\cdot (1+\delta_{i})< 1$ holds for each $\tau_i \in \tau$, and $U_{sum}\leq m-L$ holds where $L$ is defined in Lemma \ref{lem:proof1}.
\end{theorem}

\paragraph*{\textbf{Compared to the suspension-oblivious density test}}Density test \cite{goossens2003priority} is a well-known schedulability test originally designed for HRT task systems with no suspensions. The following theorem states the density test. 

\begin{theorem}
\label{theorem:dtest}
\cite{goossens2003priority} A HRT task system can be successfully scheduled by GEDF on $m$ identical processors, provided by $U_{sum} \leq m-(m-1) \cdot U_{max}$, where $U_{max}= max\{U_{1},...,U_n\}$.
\end{theorem}

By applying suspension-oblivious approach(i.e., treating all suspension as computation) to the density test, we can obtain the suspension-oblivious density test, which is he only existing utilization-based test with polynomial time complexity that can handle HRT suspending task systems. Theorem \ref{theorem:sdtest} states the suspension-oblivious density test.

\begin{theorem}
\label{theorem:sdtest}
A HRT suspending task system can be successfully scheduled by GEDF on $m$ identical processors, provided by $U_{sum} \leq m-(m-1) \cdot Z_{max}-V_{sum}$, where $Z_{i}=U_i+V_i, Z_{max}= max\{Z_{1},...,Z_n\}$ and $Z_{sum}= \sum_{i= 1}^{n}Z_{i}$.
\end{theorem}

By comparing our schedulability test to the suspension-oblivious density test, we can see that these two tests are incomparable (i.e., do not dominate each other). However, in our schedulability test, the total utilization loss is caused by the term $L$ in which $m \cdot U_i \cdot \delta_i$ is an $O(m)$ suspension-related utilization loss. While in the suspension-oblivious density test, the total utilization loss is $(m-1) \cdot Z_{max}+V_{sum}$, which  is an $\Omega (n)$ suspension-related utilization loss.

We evaluate our schedulability test by conducting extensive experiments in Sec. \ref{sec:woexp}. In the next section, we consider the read-write task model.

\section{Supporting Read-write Application}
\label{sec:rw}

In this section, we consider supporting read-write applications. Unfortunately, our analysis technique presented in Sec. \ref{sec:wo} cannot be directly applied to the read-write task model. If a job begins with a reading phase, then the value $\delta_i$ defined in Sec. \ref{sec:sysmodel} is no longer well defined. To deal with the read-write task model, we design an I/O placement policy and a corresponding new scheduling algorithm, which enable us to completely eliminate the negative impact due to read-and write-induced suspensions.

\subsection{I/O Placement Policy}
\label{subsec:ioplacement}
As shown in Fig. \ref{fig:simpleexample}, if phases are required to be executed in a pre-defined order, then the negative impact due to read-/write-induced suspensions is fundamental. Motivated by this, we design an flexible I/O placement policy that allows the  scheduler to decide when to compute and suspend within each job's execution window. This resulting desirable property is called \textit{flexible suspension pattern}.

To achieve the flexible suspension pattern, our I/O placement policy let job $\tau_{i,j-1}$ help $\tau_{i,j}$ perform its reading phase, and let job $\tau_{i,j+1}$ help $\tau_{i,j}$ perform its writing phase. Let $\overline{\tau}$ and $\tau$  denote the task system using our I/O placement policy and the original task system, respectively. For each task $\overline{\tau}_i$, a pre-fetching job $\overline{\tau}_{i,0}$ executes the reading phase of $\tau_{i,1}$; job $\overline{\tau}_{i,1}$ contains the computing phase of $\tau_{i,1}$ and the reading phase of $\tau_{i,2}$. For $j>2$, $\overline{\tau}_{i,j}$ contains the writing phase of $\tau_{i,j-1}$, the computing phase of $\tau_{i,j}$ and the reading phase of $\tau_{i,j+1}$. The following example illustrates our I/O placement policy.
\begin{figure}[t]
	\begin{center}
	\includegraphics[width=3.3in]{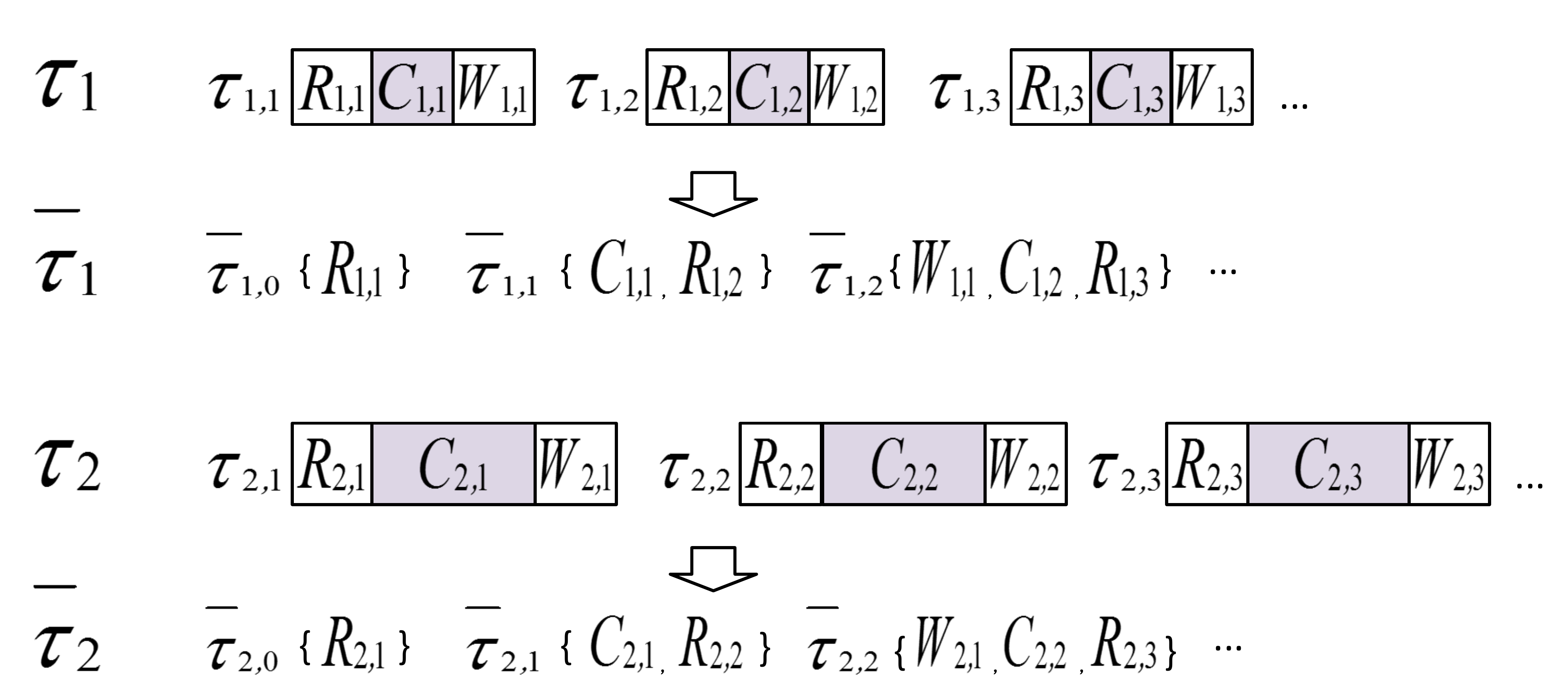} 
	\end{center} 
\vspace{-5mm}
\caption{I/O placements}
\vspace{-3mm}
\label{fig:ioplacement}
\end{figure}

\paragraph*{\textbf{Example}}Consider a read-write task system containing two tasks $\tau_1$ and $\tau_2$ on a uniprocessor platform. Each released job of both tasks reads data from disk, performs computation on processor, and writes the results back to disk.

Fig. \ref{fig:ioplacement}, shows the transformed task system $\overline{\tau}$  using our proposed flexible I/O placement policy. After the transformation, the suspension phases and the computation phases of the same job become independent. For example, for job $\overline{\tau}_{1,2}$ shown in Fig. \ref{fig:ioplacement}, since the writing phase $W_{1,1}$ writes the output of the computing phase $C_{1,1}$ to disk, it has no dependency with the computing phase $C_{2,1}$. 

Note that the reordering process happens in the application programming phase. It has been well designed before running and thus will not incur the locality loss of data at runtime. Based on this I/O placement policy, we design the following GEDF-based scheduling algorithm.

\begin{figure}[]
  \centering
  \subfloat[Task system using the original  I/O placement under GEDF]{\label{fig:ca}\includegraphics[width=0.5\textwidth]{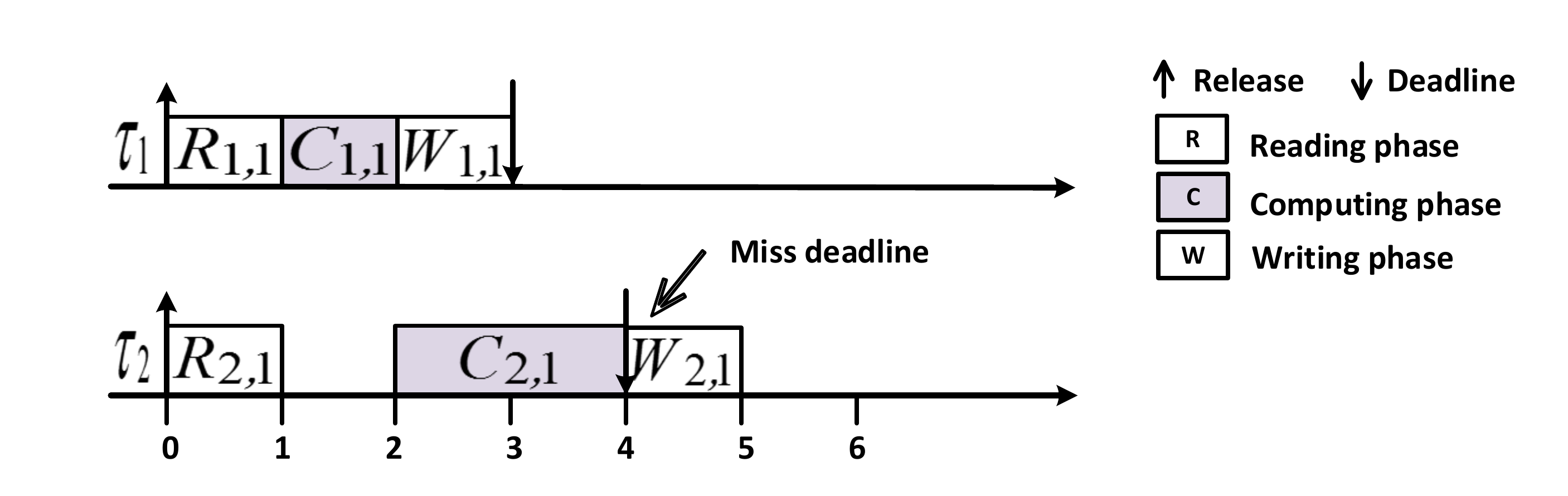}} \hspace{0mm}
  \subfloat[Task system using the original  I/O placement under GEDF-R/W]{\label{fig:cb}\includegraphics[width=0.5\textwidth]{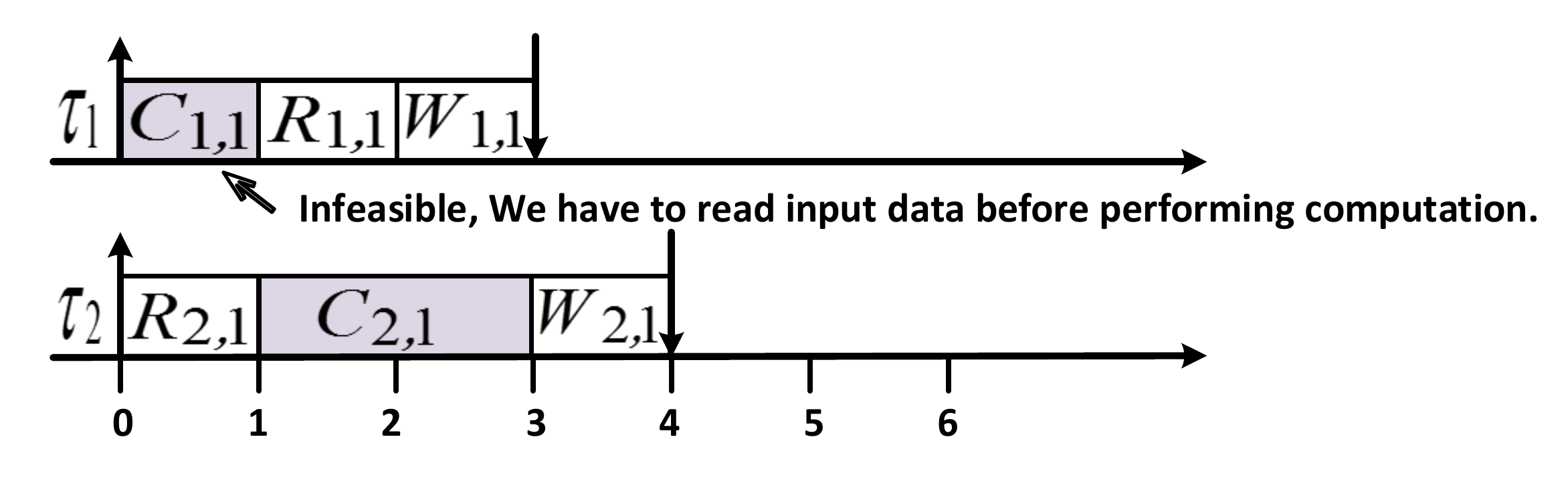}} \hspace{0mm}
  \subfloat[Task system using our I/O placement under GEDF-R/W]{\label{fig:cc}\includegraphics[width=0.5\textwidth]{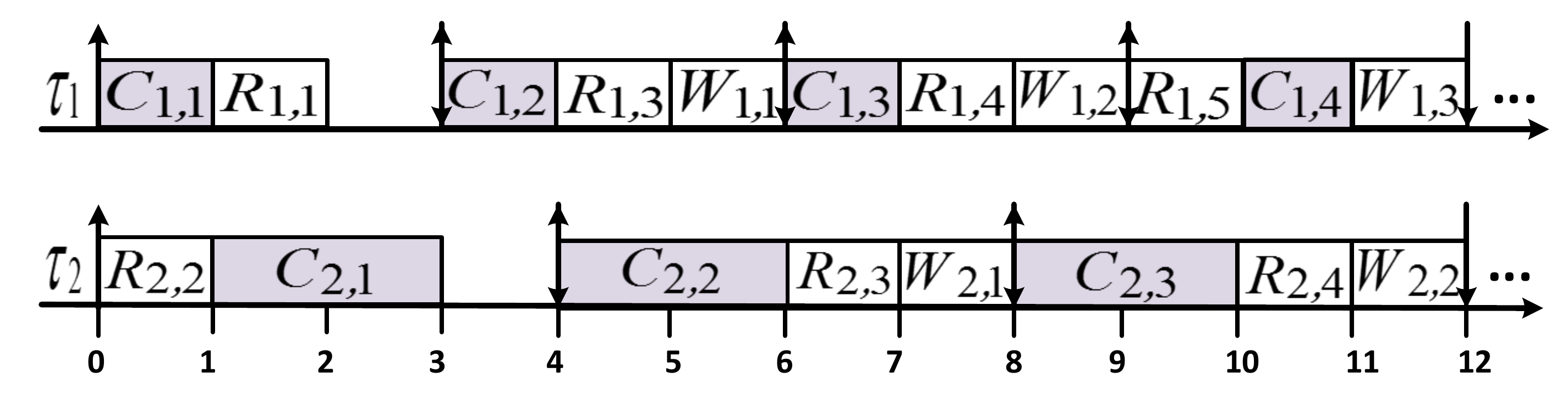}} \hspace{0mm} 
  \caption{GEDF-R/W schedules} \vspace{-2mm} \normalsize
  \label{fig:case_sche}
\end{figure}

\subsection{GEDF-R/W Scheduling Algorithm}
\label{subsec:EEDF}
Our scheduling algorithm, denoted GEDF-R/W, extends GEDF in a way such that the scheduler decides when to perform computation and suspension. GEDF-R/W is defined as follows.

\begin{definition}
\label{def:GEDF-R/W}
At each time instant, GEDF-R/W selects $m$ comp-pending jobs with the earliest deadlines for computation. If an comp-pending job $\tau_{i,j}$ is comp-preempted, $\tau_{i,j}$ will perform suspension if it has not finished all its suspensions. 
\end{definition}

We use the example task system in Sec. \ref{subsec:ioplacement} to illustrate GEDF-R/W scheduling algorithm. In Fig. \ref{fig:case_sche}(a), we can easily see that $\tau$ is not schedulable under GEDF even the total system utilization is low. In Fig. \ref{fig:case_sche}(b), we try to apply GEDF-R/W to the original task system $\tau$. However GEDF-R/W scheduling is infeasible without the flexible suspension pattern. Fig. \ref{fig:case_sche}(c) shows that $\overline{\tau}$ can be successfully scheduled by GEDF-R/W. 


\subsection{Scheduling Analysis}
\label{subsec:rwanalysis}

In this section, we analyze the GEDF-R/W schedule of read-write task systems using our I/O placement policy.

\paragraph*{\textbf{LAG and SLAG}}In Sec. \ref{sec:wo}, we introduce the lag-based technique and the related definitions. In this section, $A(\tau_{i,j}$ $, t_1, t_2, S)$, $lag(\tau_i, t, S)$ and $LAG(\tau, t, S)$ are defined in the same manner corresponding to the GEDF-R/W schedule $S$ and the PS schedule $PS$.

For a write-only task, a job is finished if and only if it has finished all its computation because it ends with a computing phase. However, this is not true for read-write tasks, where the last phase finished in a job of a read-write task could be a suspension phase. Intuitively, the PS schedule provides a good means to track the progress on computation performed in the GEDF-R/W schedule. But only using the PS schedule is insufficient to track the progress on suspensions performed in the GEDF-R/W schedule, in which case it is hard to check whether a job of a read-write task misses its deadline (because a job can still miss its deadline while completing all the computation). Therefore, to deal with read-write task model, we define the following perfect schedule for suspensions denoted as the SPS schedule. In a SPS schedule $SPS$, each read-write task $\tau_i$ suspends with a speed equal to $V_i$ when it is sus-pending (which ensures that each of its jobs finishes its suspension exactly at its deadline). We use the task system in Fig. \ref{fig:simpleexample} to illustrate the $SPS$ schedule as shown in Fig. \ref{fig:SPSschedule}. In the SPS schedule, the two tasks perform suspensions with a speed equal to their corresponding suspension ratios 2/3. Similar to $PS$, $SPS$ is not a real schedule and is only used for analysis purposes to keep track of the suspension.

\begin{figure}[th]
	\begin{center}
	\includegraphics[width=2.5in]{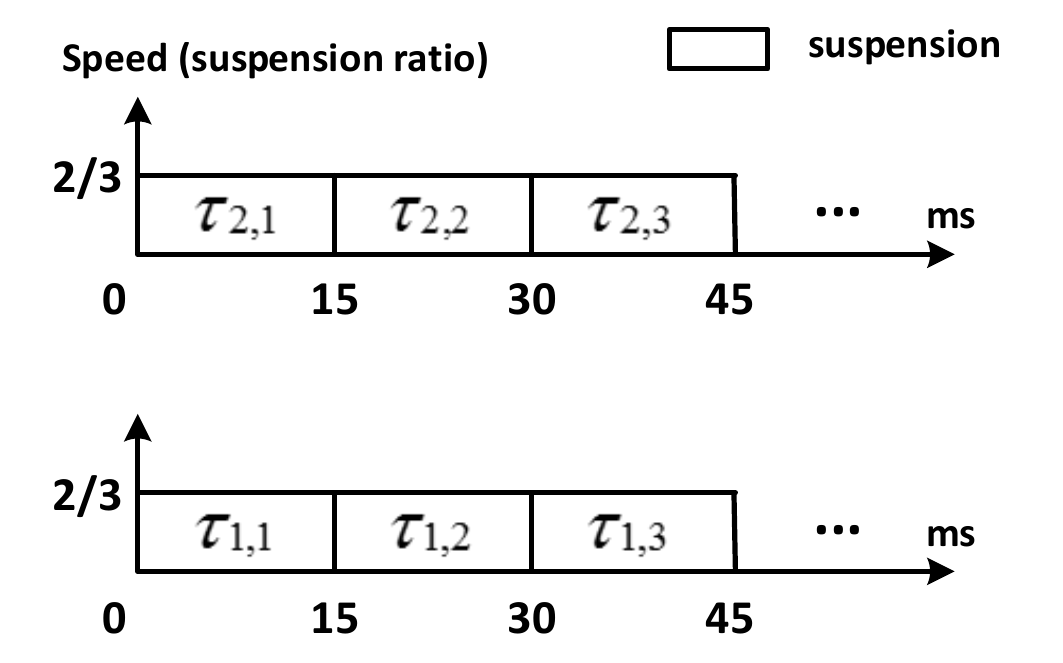} 
	\end{center} 
\vspace{-5mm}
\caption{Example SPS schedule}
\vspace{-3mm}
\label{fig:SPSschedule}
\end{figure}

Let $slag(\tau_{i,j},t,S)$ denote the difference of suspensions done by $\tau_{i,j}$ in $SPS$  and $S$. Next, we analyze the schedulability for a read-write task system on $m$ identical processors under GEDF-R/W.

\begin{lemma}
\label{lem:pre_analysis}
If job $\tau_{i,j}$ misses its deadline at $d_{i,j}$, then   
\begin{equation}
\label{eq:pre}
lag(\tau_{i,j}, d_{i,j}, S) > 0.                   
\end{equation}
\end{lemma}

\begin{proof}
By the definitions of $lag(\tau_{i,j}, d_{i,j}, S)$ and $slag(\tau_{i,j}$ $, d_{i,j}, S)$, they cannot be negative at $d_{i,j}$. And because $\tau_{i,j}$ has not finished all its suspensions or computation at $d_{i,j}$, at least one of them must be positive. Thus,
\begin{equation}
\label{eq:as_lag1_2}
lag(\tau_{i,j}, d_{i,j}, S)+ slag(\tau_{i,j}, d_{i,j}, S) > 0.                   
\end{equation}

Suppose that $\tau_{i,j}$ has not finished its suspensions at $d_{i,j}$. In this case, there are three kinds of unit intervals to be considered in  $[r_{i,j},d_{i,j})$. Let $\gamma_1$ denote the set of unit intervals in which $\tau_{i,j}$ is comp-preempted and suspends; let $\gamma_2$ denote the set of unit intervals in which $\tau_{i,j}$ is not comp-preempted and suspends; let $\gamma_3$ denote the set of unit intervals in which $\tau_{i,j}$ is not comp-preempted and computes. Thus $\gamma_1\cup \gamma_2 \cup \gamma_3= [r_{k,l}, d_{i,j})$ and they are pairwise disjoint. Let $L_1$, $L_2$ and $L_3$ be the length of each set, respectively.  According to GEDF-R/W, $\tau_{i,j}$ must suspend in all unit intervals in $\gamma_1$. Thus,
\begin{eqnarray*}
&& lag(\tau_{i,j}, d_{i,j}, S) \\
&=& U_i \cdot (L_1+L_2+L_3)- L_3 ,\\
and && ~\\    
& & slag(\tau_{i,j}, d_{i,j}, S)  \\
&=& V_i \cdot (L_1+L_2+L_3)- (L_1+L_2),\\
and && ~\\
&& lag(\tau_{i,j}, d_{i,j}, S)+ slag(\tau_{i,j}, d_{i,j}, S)  \\
&=& (U_i+ V_i- 1)\cdot (L_1+L_2+L_3) \\
&&{\lbrace \rm{by}~(\ref{eq:req1_2}) \rbrace} \\
&\leq& 0,
\end{eqnarray*}   
which violates Eq. (\ref{eq:as_lag1_2}). Therefore, $\tau_{i,j}$ must have finished all of its suspensions but have not finished its computation, which implies $lag(\tau_{i,j}, d_{i,j}, S)>0$.

Lemma \ref{lem:pre_analysis} thus follows. Intuitively, Lemma \ref{lem:pre_analysis} implies that when a job misses its deadline it is necessary that it has not finished its computation, which can be used to derive an necessary condition for deadline miss, as shown in Lemma. \ref{proof2}.
\end{proof}

\begin{lemma}
\label{proof2}
If
\begin{equation}
\label{eq:as_sum_2}
U_{sum} \leq m-(m-1) \cdot U_{max},                   
\end{equation}
then no job misses its deadline under GEDF-R/W.
\end{lemma}

\begin{proof}
Due to the limitation of space, we present this proof in the appendix.
\end{proof}

Lemma \ref{proof2} implies the following schedulability test.

\begin{theorem}
\label{theorem:RWIStest}
A HRT read-write task system using our I/O placement policy can be successfully scheduled under GEDF-R/W on $m$  identical processors, provided by $U_{sum} \leq m-(m-1)\cdot U_{max}$.
\end{theorem}

This schedulability test is identical to the density test shown in Theorem \ref{theorem:dtest} for ordinary task systems with no suspensions, which implies the negative impact of suspension has been completely eliminated.

\section{Case Study}
\label{sec:case}
In this section, we show the feasibility of our I/O placement policy and the corresponding GEDF-R/W scheduling via a case study implementation. We also evaluate the scheduling performance with respect to the response time bound of tasks. We programmed the real-time matrix calculation read-write applications, which read matrix from disk, perform the matrix calculation, and write the result to disk. For conciseness, let us denote the application programmed using our I/O placement policy as \textit{our application} and the application programmed using the original I/O placement policy as \textit{original application}.

\subsection{Implementation}
Our case study was conducted on an ASUS machine with a two-core CPU running at 3.40GHz. In order to get noticeable response times we used matrices with size 500*500. First, we generated 5000 matrices stored in the disk and the elements in the matrices were randomly generated integers using a uniform distribution $[0, 9]$. We used GEDF to schedule the original application and used GEDF-R/W to schedule our application. We recorded the response time of the first 400 jobs of each task. In our experiments, $100 ms$ is the unit interval of computation and suspension. 

We conducted experiments for two cases: (1) two tasks on a uniprocessor; (2) three tasks on two processors. In case (1), the read-write application has two tasks $\tau_1$ and $\tau_2$ and the function of each task in shown in Fig. \ref{fig:case_func}. As the figure shows, the reading of $\tau_1$ contains 3 unit intervals that reads three matrix $A$, $B$ and $C$ from disk, respectively; the computing phase of $\tau_1$ contains 2 unit intervals that perform two multiplication operations respectively; the writing phase of $\tau_1$ contains 1 unit interval that writes the resulting matrix to disk. The second task $\tau_2$ in case (1) has the similar work mode as shown in Fig. \ref{fig:case_func}. We also pre-conducted an experiment to estimate the length of each kind of phases. Reading one matrix from disk or writing one matrix back to disk consumes less than $100 ms$ in the worst-case and each matrix operation consumes less than $200 ms$ in the worst-case. Thus, we set the periods of $\tau_1$ and $\tau_2$ as $950 ms$ and $1250 ms$, respectively. In case (2), we ran three tasks $\overline{\tau}_1$, $\overline{\tau}_2$ and $\overline{\tau}_3$ on two processors where $\overline{\tau}_1$ and $\overline{\tau}_3$ are identical to $\tau_1$ in case (1) and $\overline{\tau}_2$ is identical to $\tau_2$ in case (1).  




\begin{figure}[h]
	\begin{center}
	\includegraphics[width=3in]{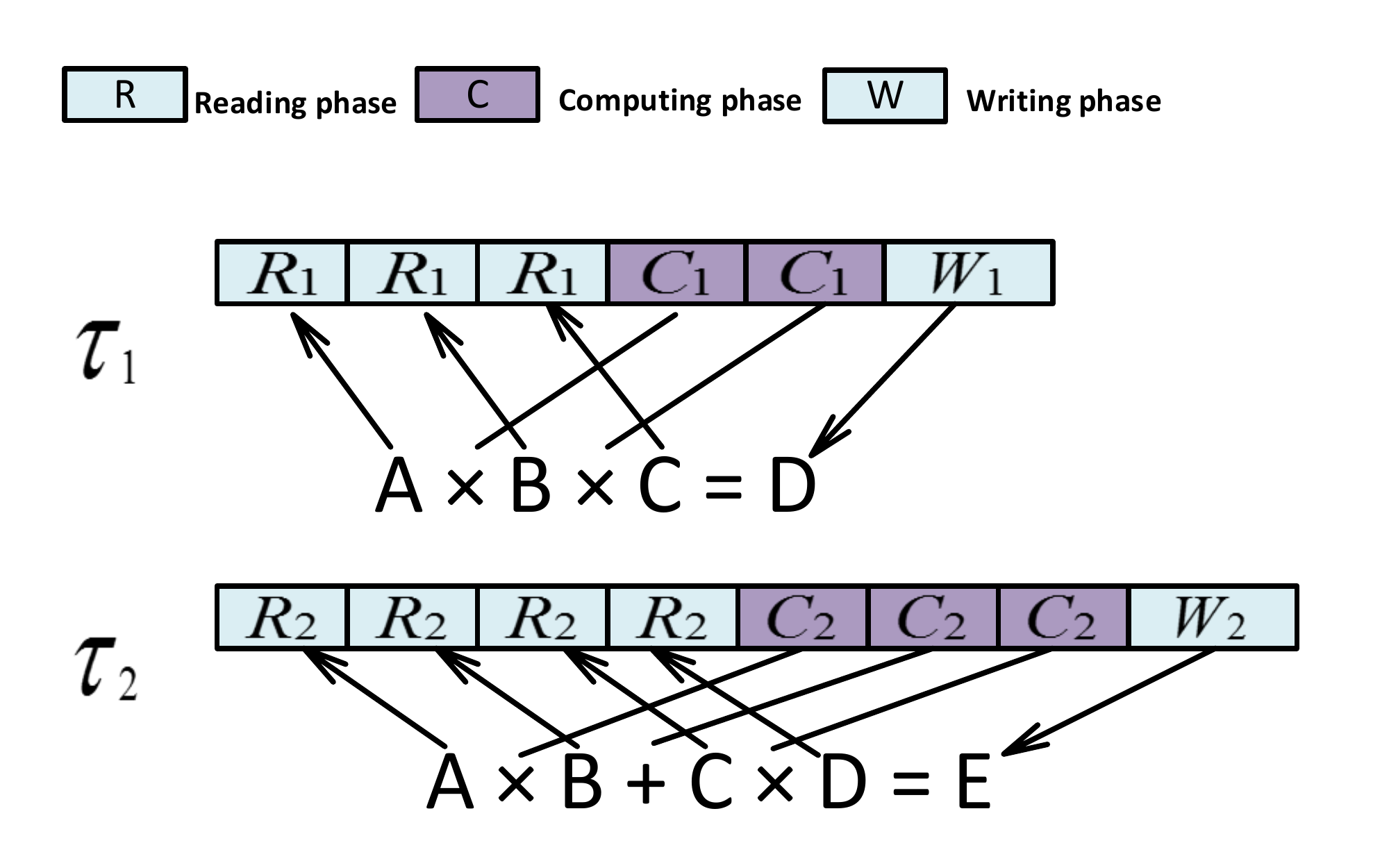} 
	\end{center} 
\vspace{-5mm}
\caption{Matrix calculation}
\vspace{-5mm}
\label{fig:case_func}
\end{figure}

\subsection{Performance Evaluation}
The performance of each cases is shown in Fig. \ref{fig:case}. The $x$-axis denotes the job number and the $y$-axis denotes the response time of each job.

In case (1), the total utilization of $\tau_1$ and $\tau_2$ is about $0.9$ ($400/950+600/1250$) which does not exceed $1$.  According to the analysis in Sec. \ref{subsec:rwanalysis} the response time of $\tau_1$ and $\tau_2$ in our application should not exceed $950 ms$ and $1250 ms$. However, our analysis has not taken the overhead due to job migration into consideration. From Fig. \ref{fig:case} we can see, in practice,  the response time of some jobs of $\tau_1$ in our application is about $50 ms$ lager than the theoretical estimation. However, compare to GEDF, GEDF-R/W performed considerably better with respect to reducing response time. Moreover, as shown in Figs. \ref{fig:case}(a) and (b), GEDF-R/W is able to achieve bounded response time while the response time under GEDF  grows unboundedly.
\begin{figure*}[]
  \centering
  \subfloat[Case 1: Tasks in the original application]{\label{fig:11}\includegraphics[width=0.32\textwidth]{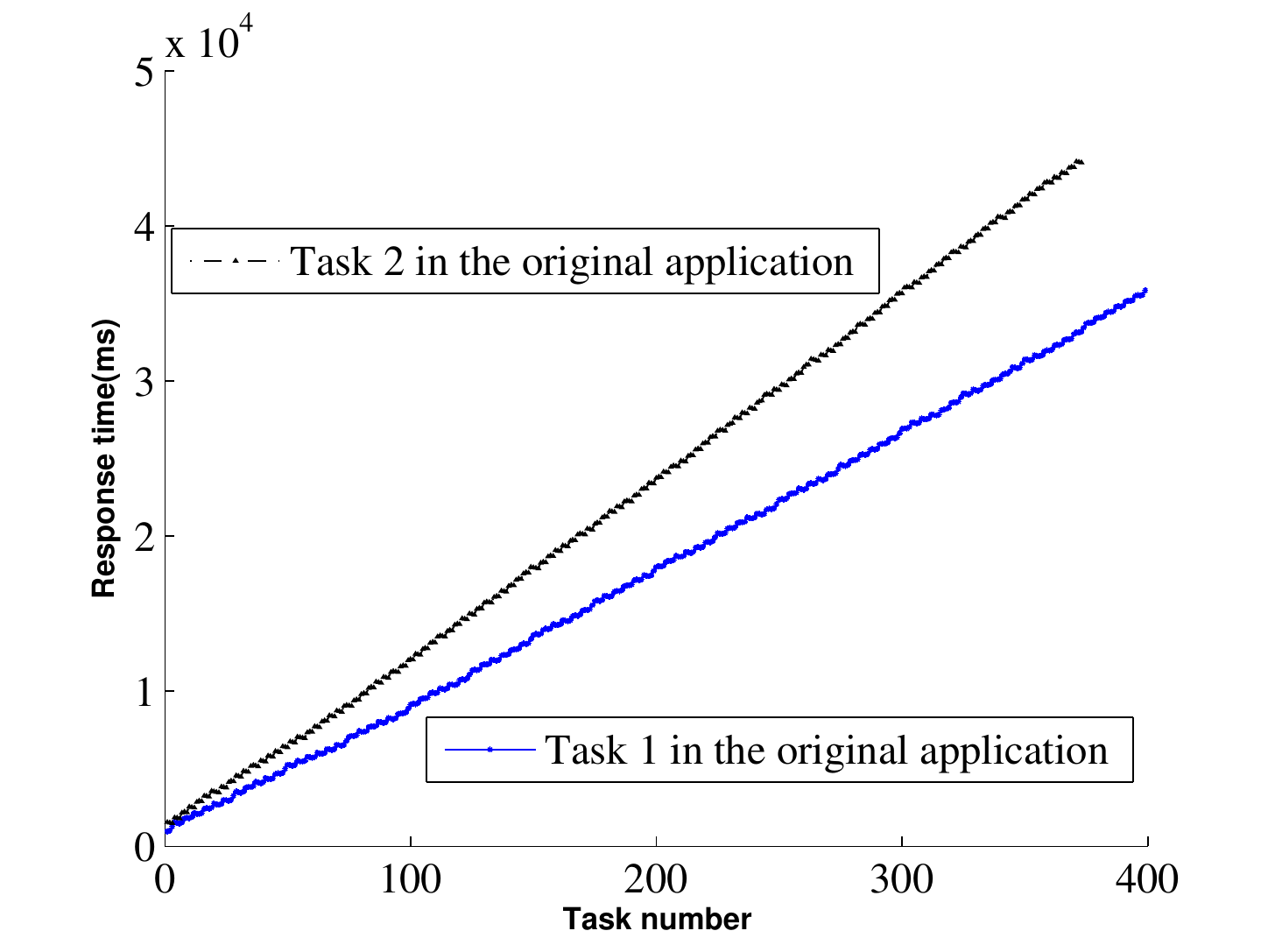}} \hspace{0mm}
  \subfloat[Case 1: Tasks in our application]{\label{fig:12}\includegraphics[width=0.32\textwidth]{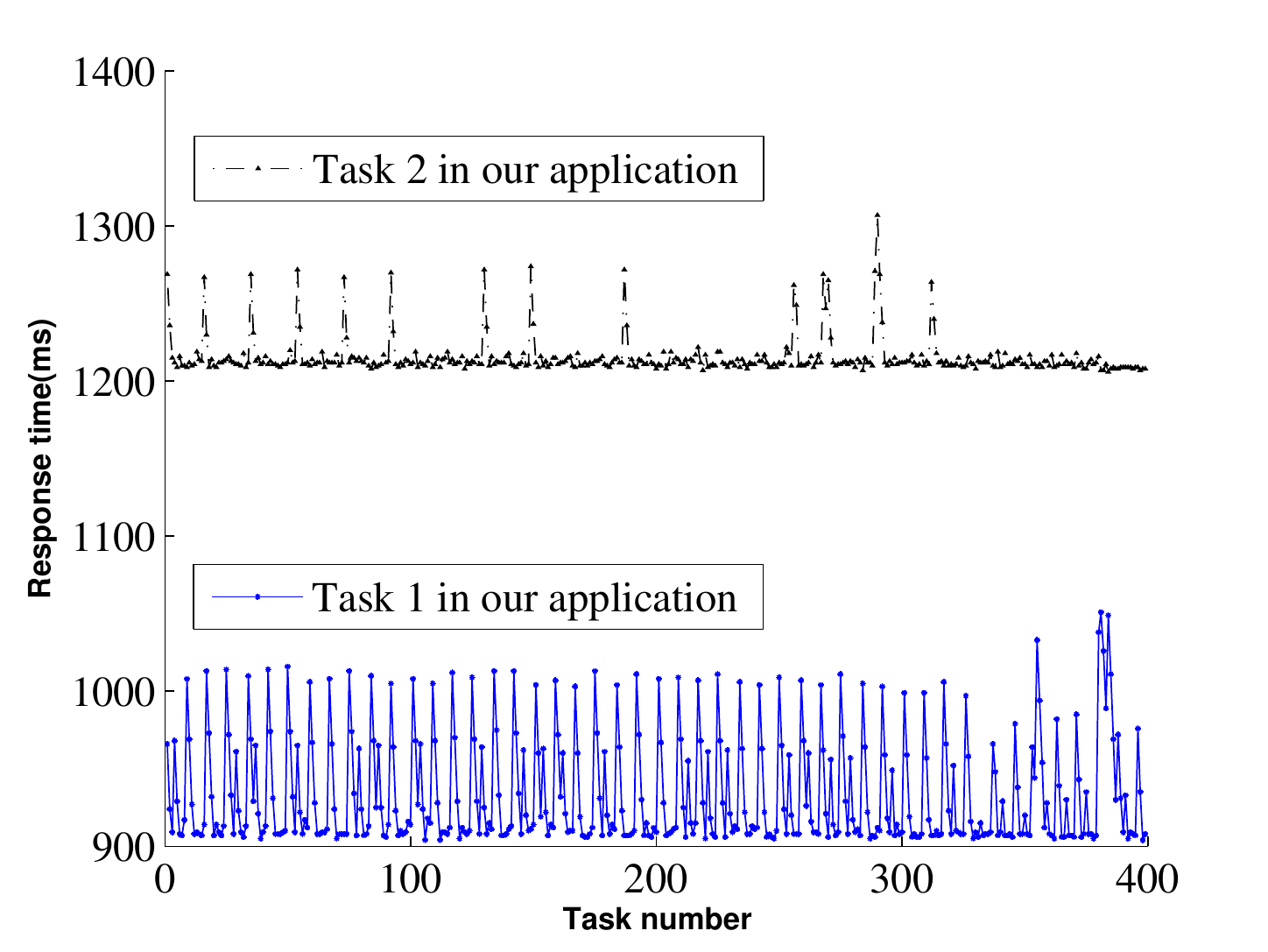}} \hspace{0mm}
  \subfloat[Case 2: Response time of $\tau_3$ and $\overline{\tau}_3$]{\label{fig:13}\includegraphics[width=0.32\textwidth]{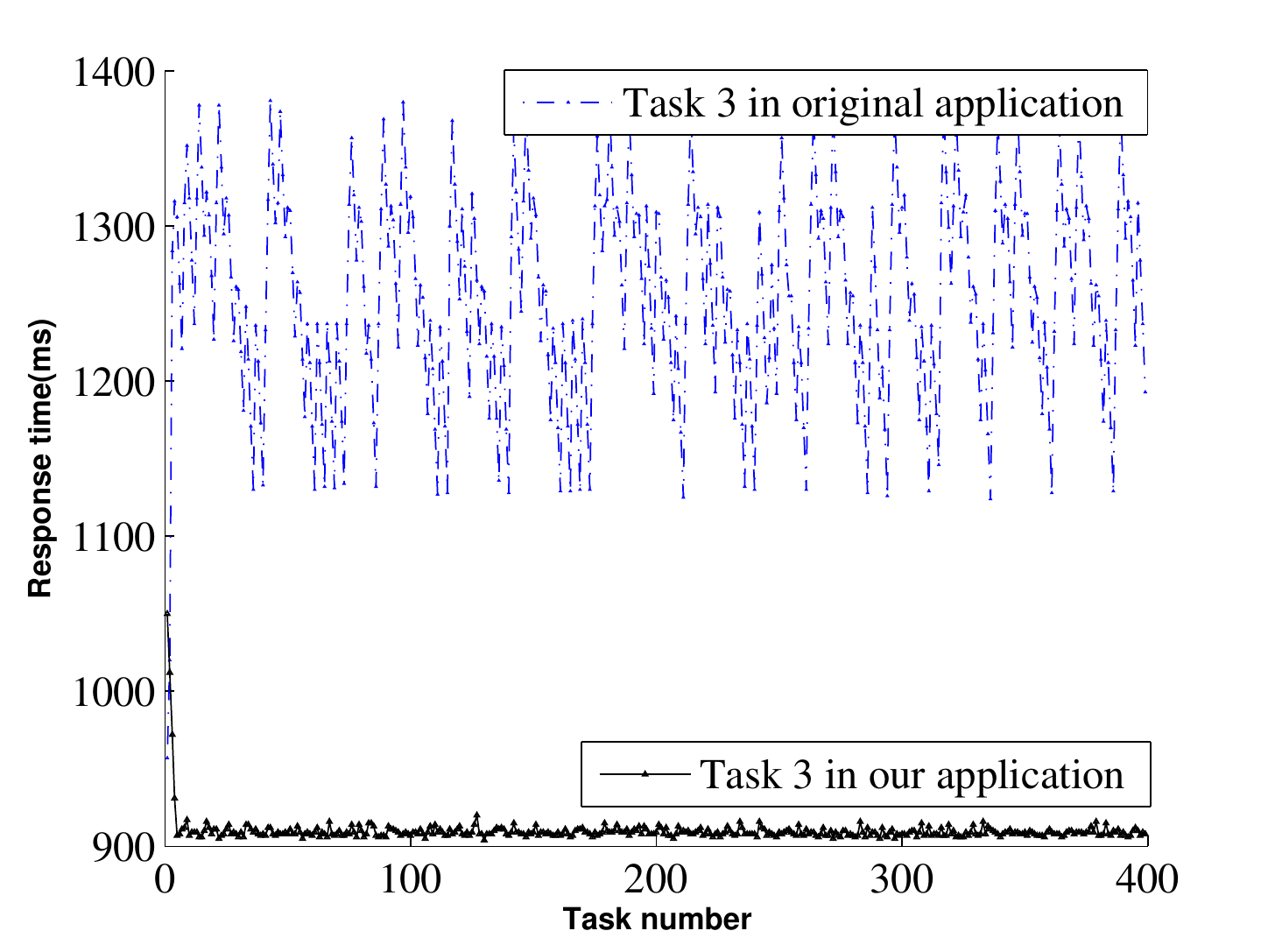}} \hspace{0mm} 
  \caption{Response time of each case} \vspace{-2mm} \normalsize
  \label{fig:case}
\end{figure*}
In case (2), the experiment shows that GEDF-R/W can significantly reduce response times in the multiprocessor case. For instance, in Fig. \ref{fig:case}(c), the response time of $\tau_3$ under GEDF varies from $1150 ms$ to $1550 ms$. While the response time of $\overline{\tau}_3$ under GEDF-R/W is merely around $920 ms$ with very slight variance.

\section{Experimental evaluation}
\label{sec:woexp}
In this section we evaluate our schedulability test stated in Theorem \ref{theorem:wotest} by experiments. Our goal is to examine how restrictive the derived schedulability test is and to compare it with suspension-oblivious density test shown in Theorem \ref{theorem:dtest}. 

\begin{figure*}[!htb]
  \centering
  
  \includegraphics[width=3.2in,height=2cm]{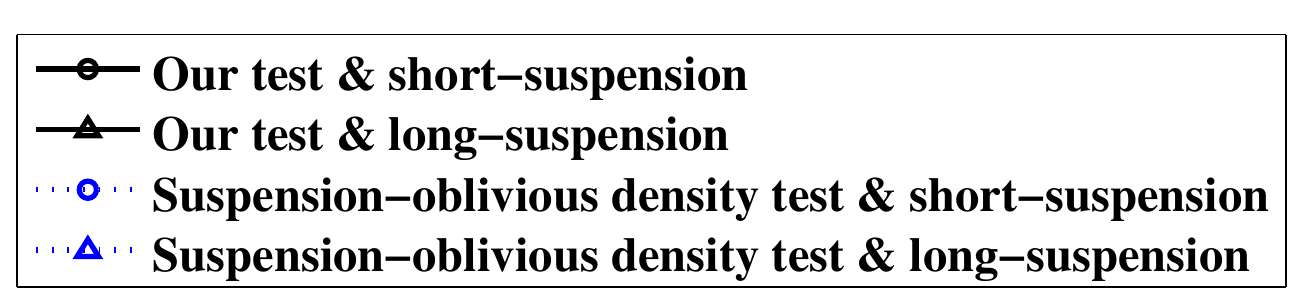} 
  
  \subfloat[Light utilization with $\alpha_i=0.9$]{\label{fig:1}\includegraphics[width=0.32\textwidth,height=3cm]{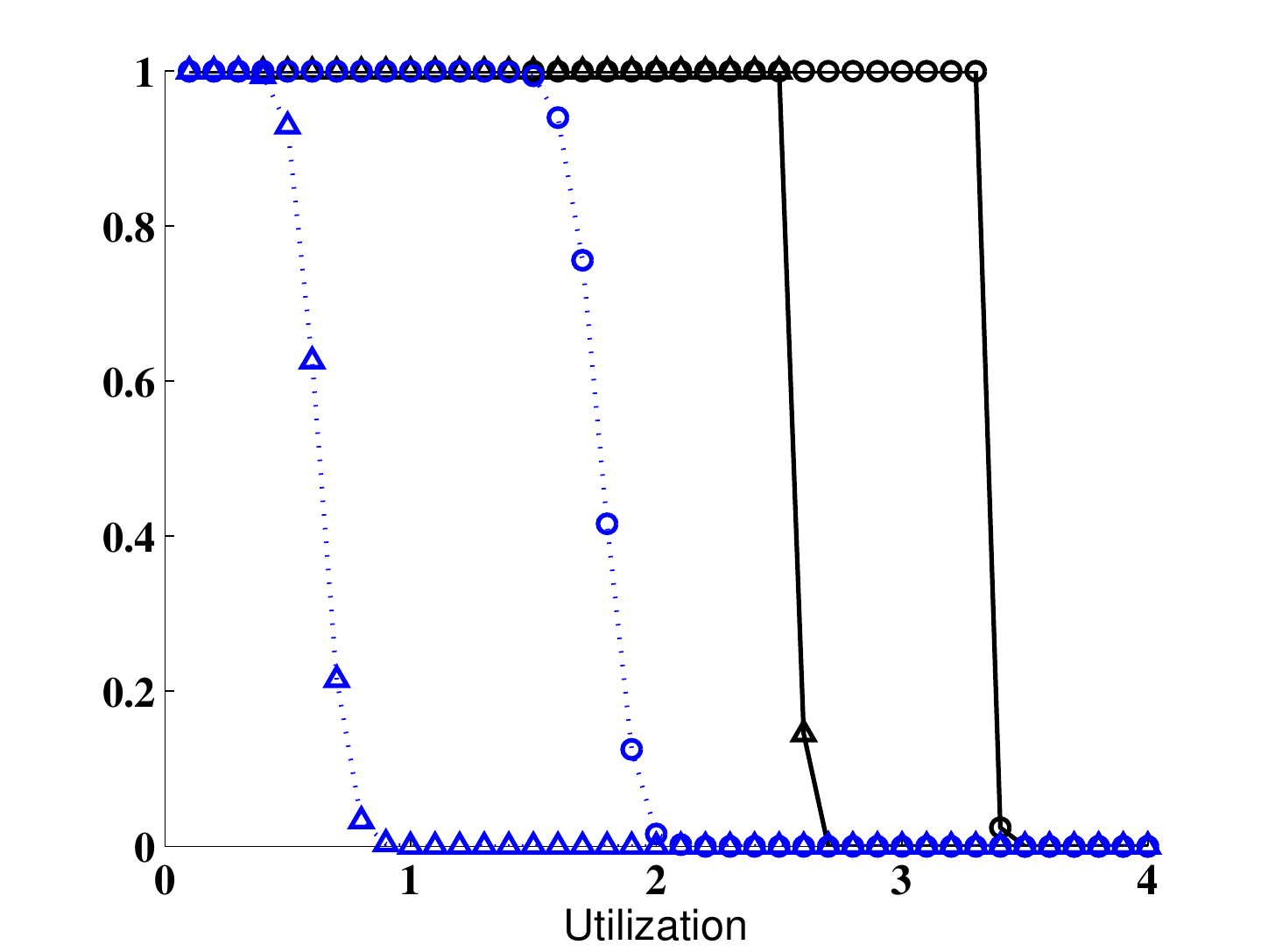}} \hspace{0mm}
  \subfloat[Medium utilization with $\alpha_i=0.9$]{\label{fig:2}\includegraphics[width=0.32\textwidth,height=3cm]{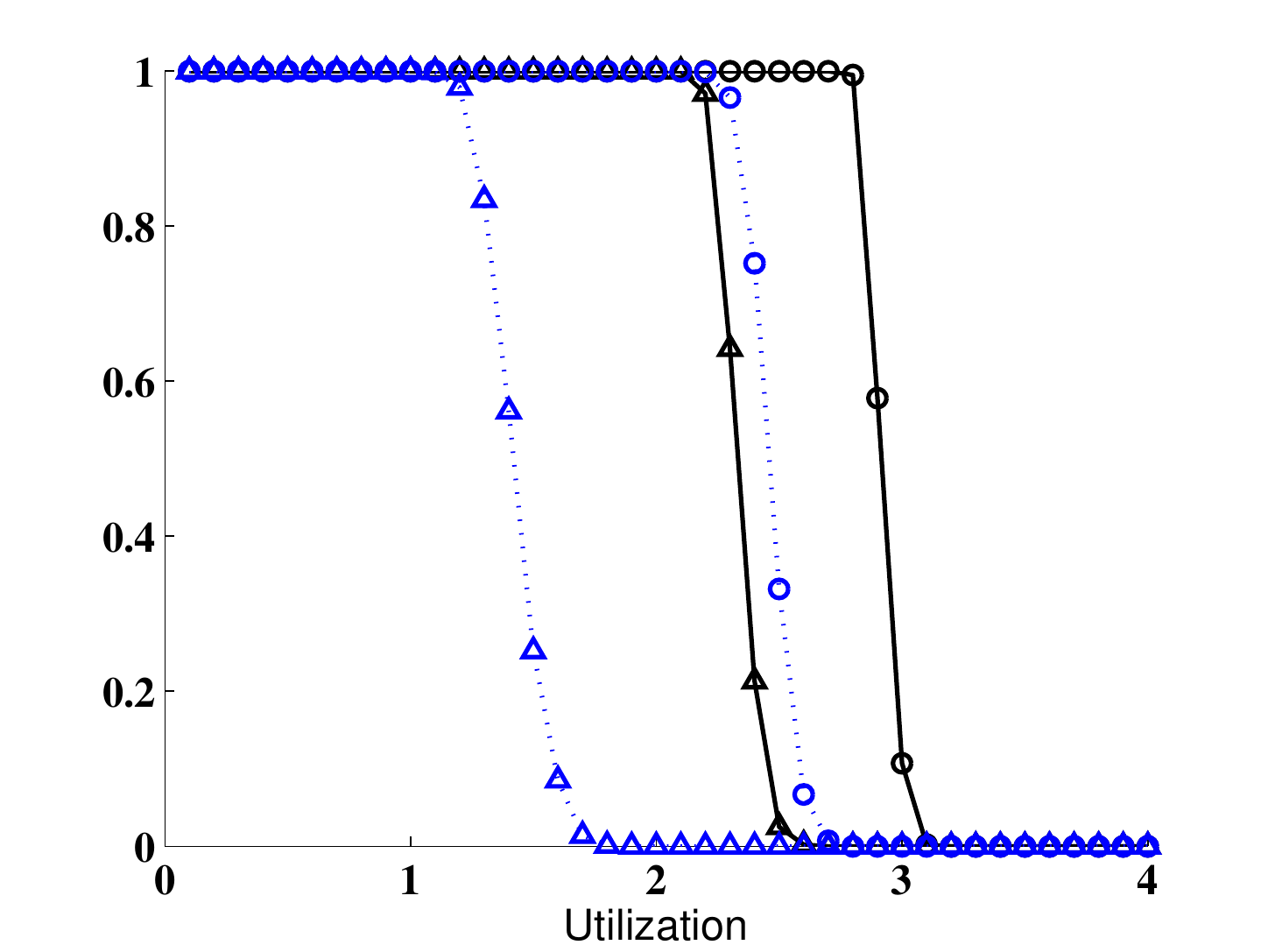}} \hspace{0mm}
  \subfloat[Heavy utilization with $\alpha_i=0.9$]{\label{fig:3}\includegraphics[width=0.32\textwidth,height=3cm]{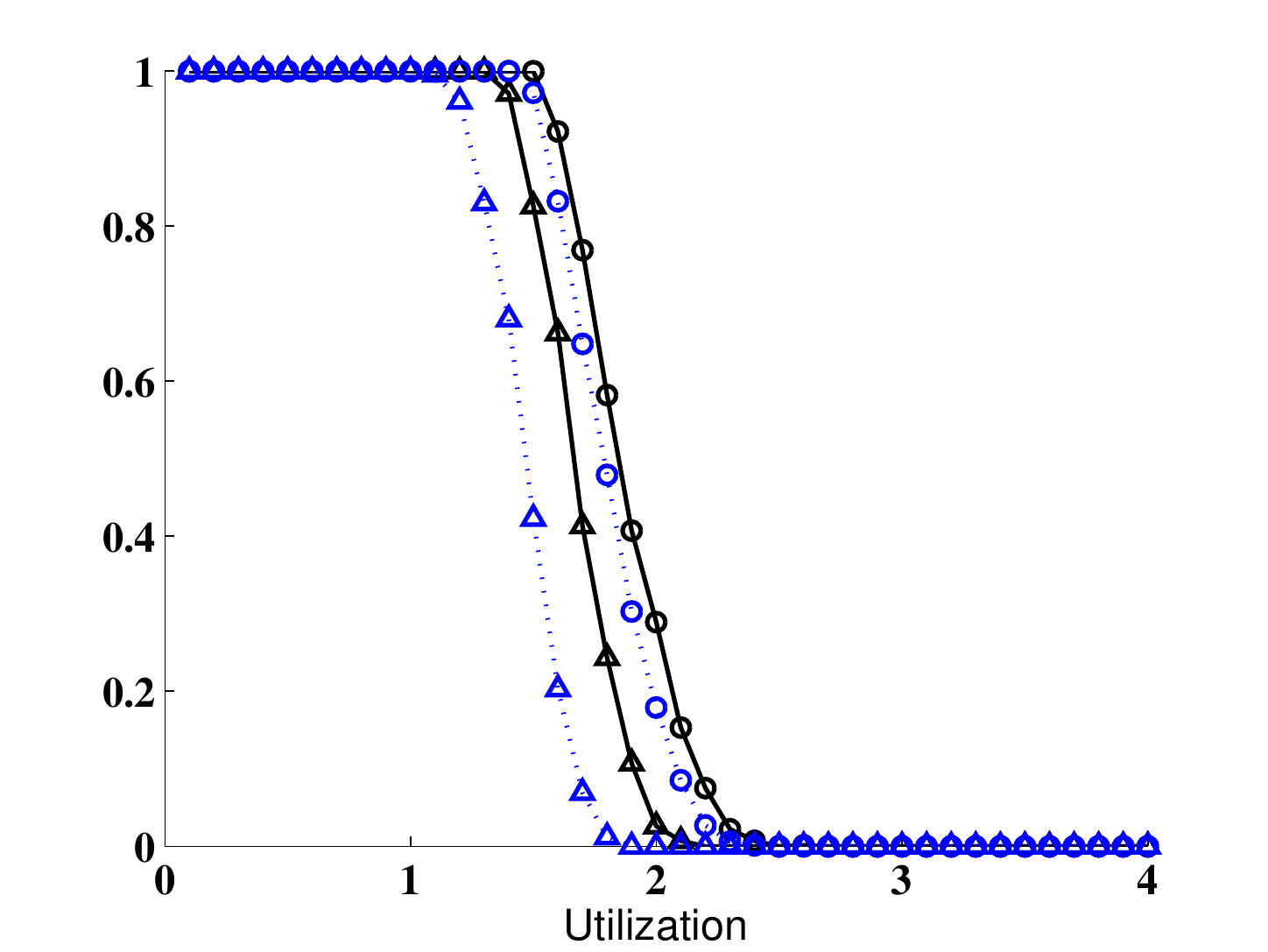}} \hspace{0mm}
  
  \subfloat[Light utilization with $\alpha_i=0.5$]{\label{fig:4}\includegraphics[width=0.32\textwidth,height=3cm]{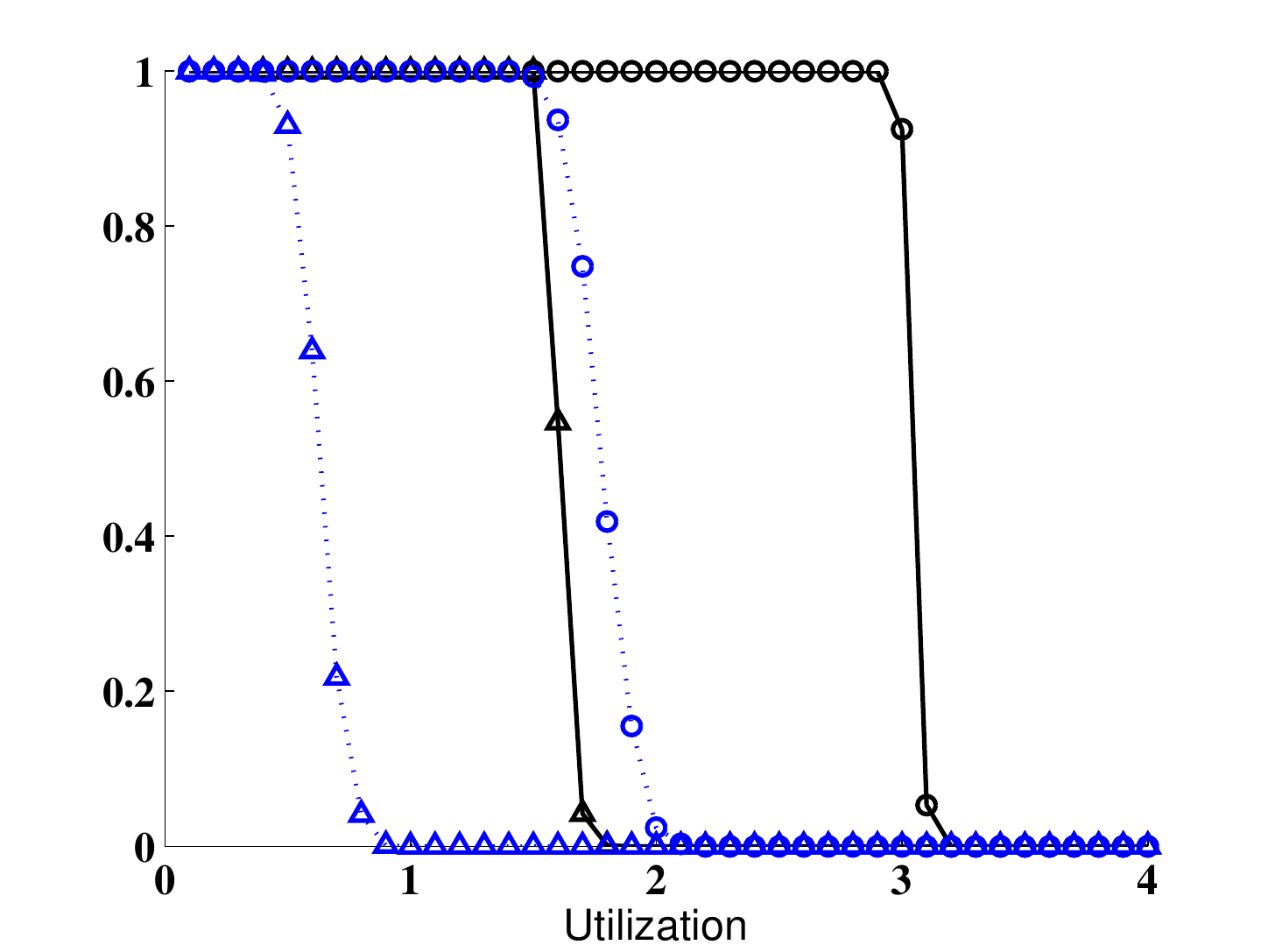}} \hspace{0mm}
  \subfloat[Medium utilization with $\alpha_i=0.5$]{\label{fig:5}\includegraphics[width=0.32\textwidth,height=3cm]{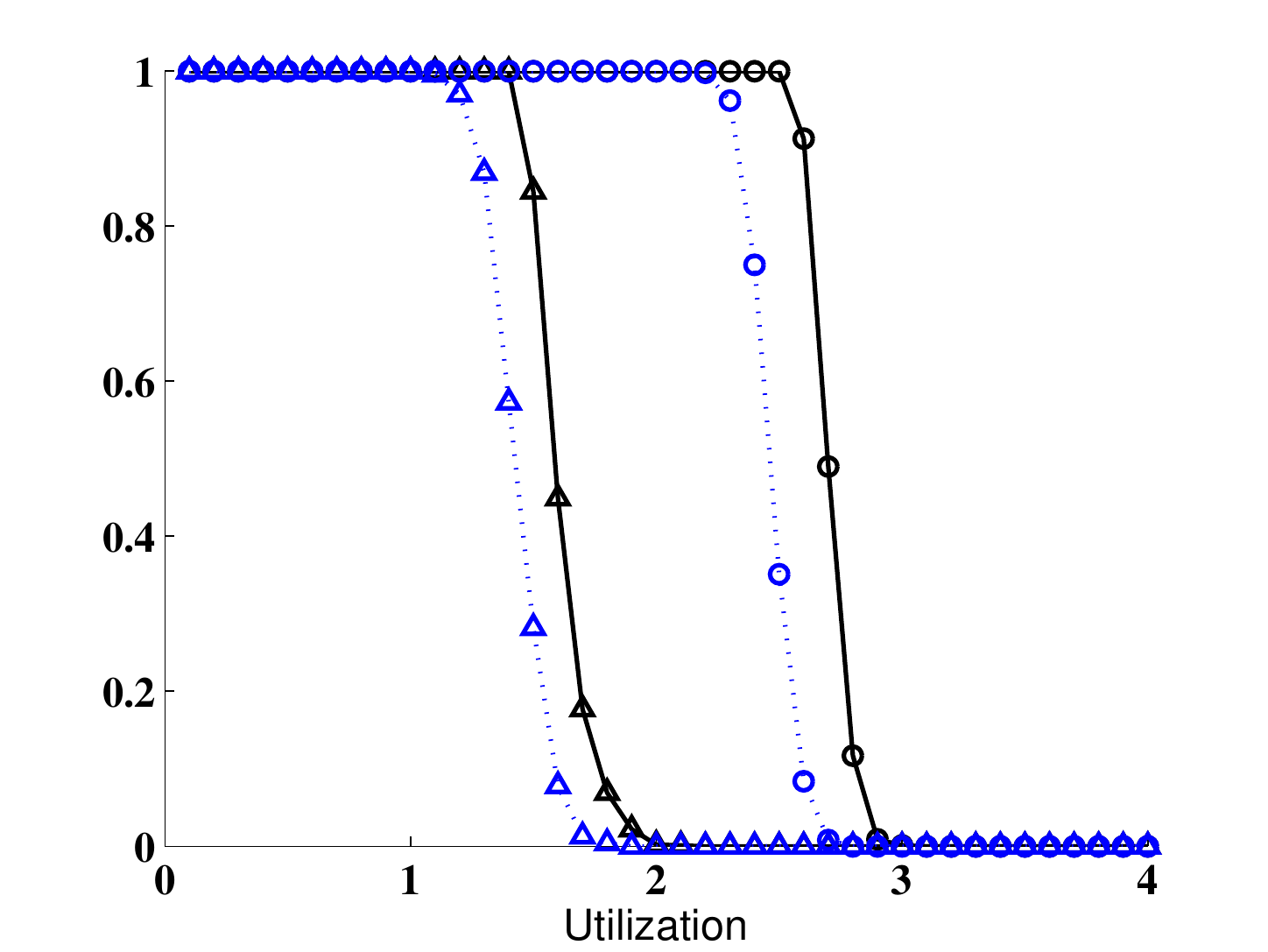}} \hspace{0mm}
  \subfloat[Heavy utilization with $\alpha_i=0.5$]{\label{fig:6}\includegraphics[width=0.32\textwidth,height=3cm]{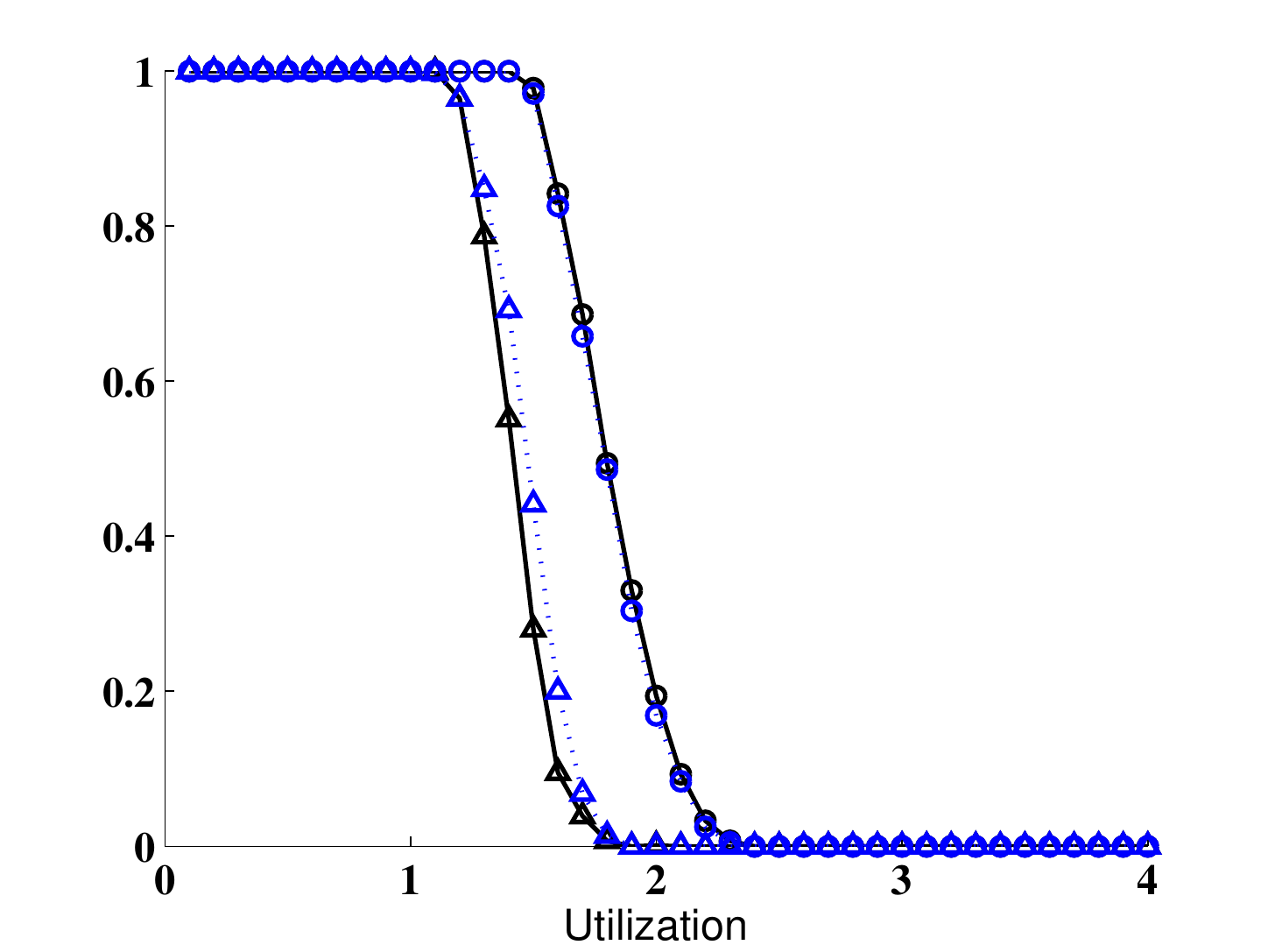}} \hspace{0mm}
  
  \subfloat[Light utilization with $\alpha_i=0.2$]{\label{fig:7}\includegraphics[width=0.32\textwidth,height=3cm]{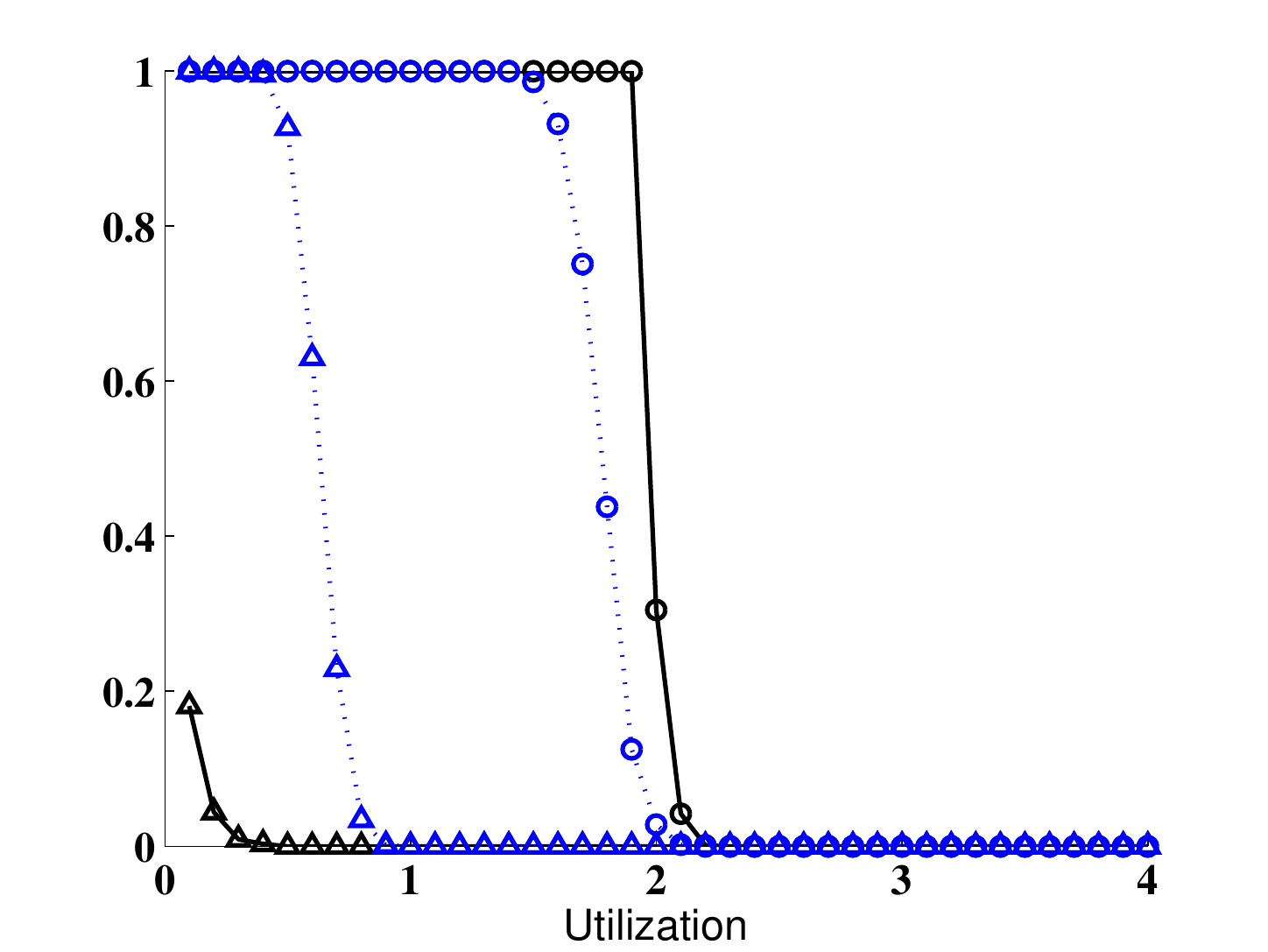}} \hspace{0mm}
  \subfloat[Medium utilization with $\alpha_i=0.2$]{\label{fig:8}\includegraphics[width=0.32\textwidth,height=3cm]{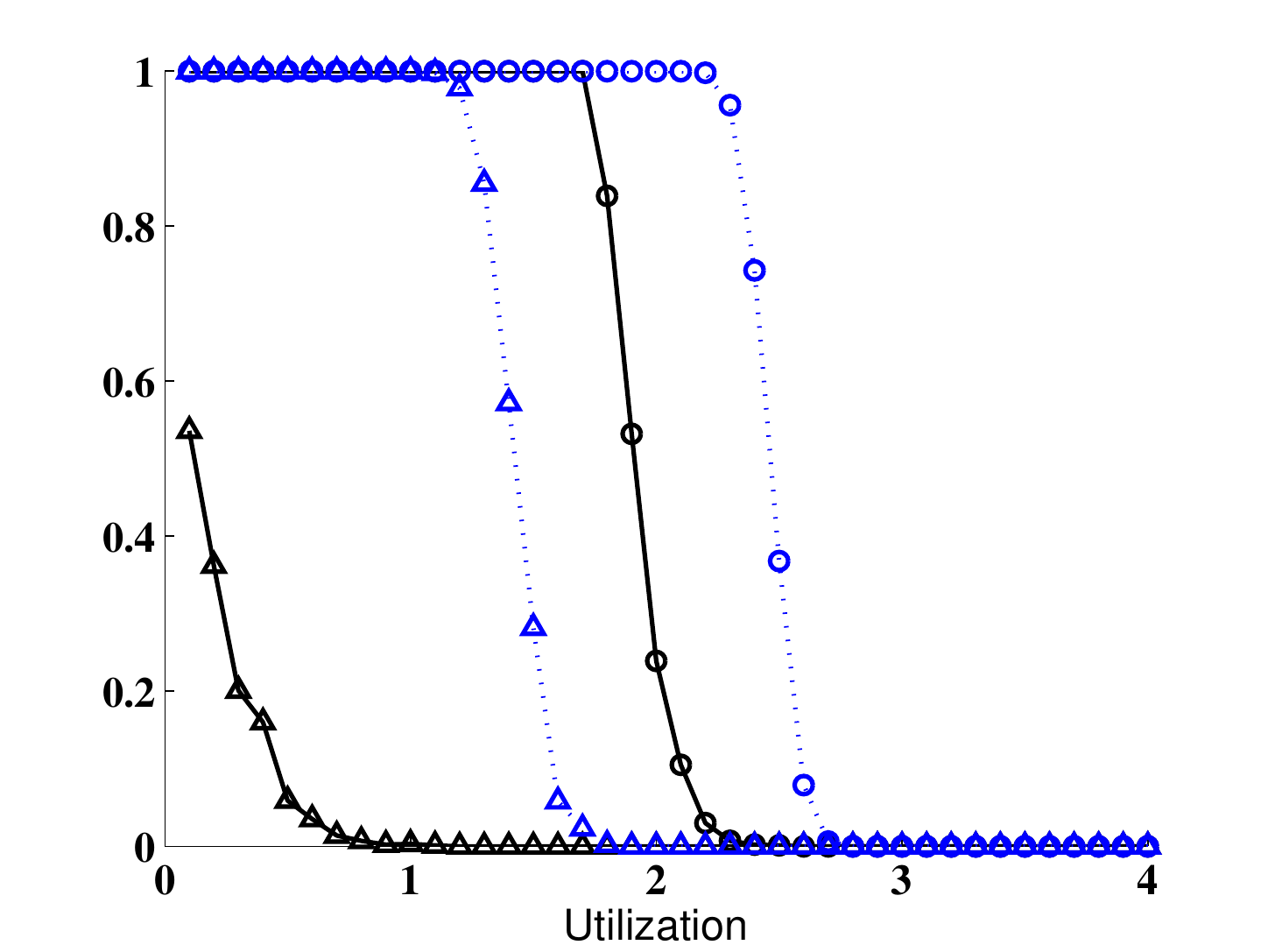}} \hspace{0mm}
  \subfloat[Heavy utilization with $\alpha_i=0.2$]{\label{fig:9}\includegraphics[width=0.32\textwidth,height=3cm]{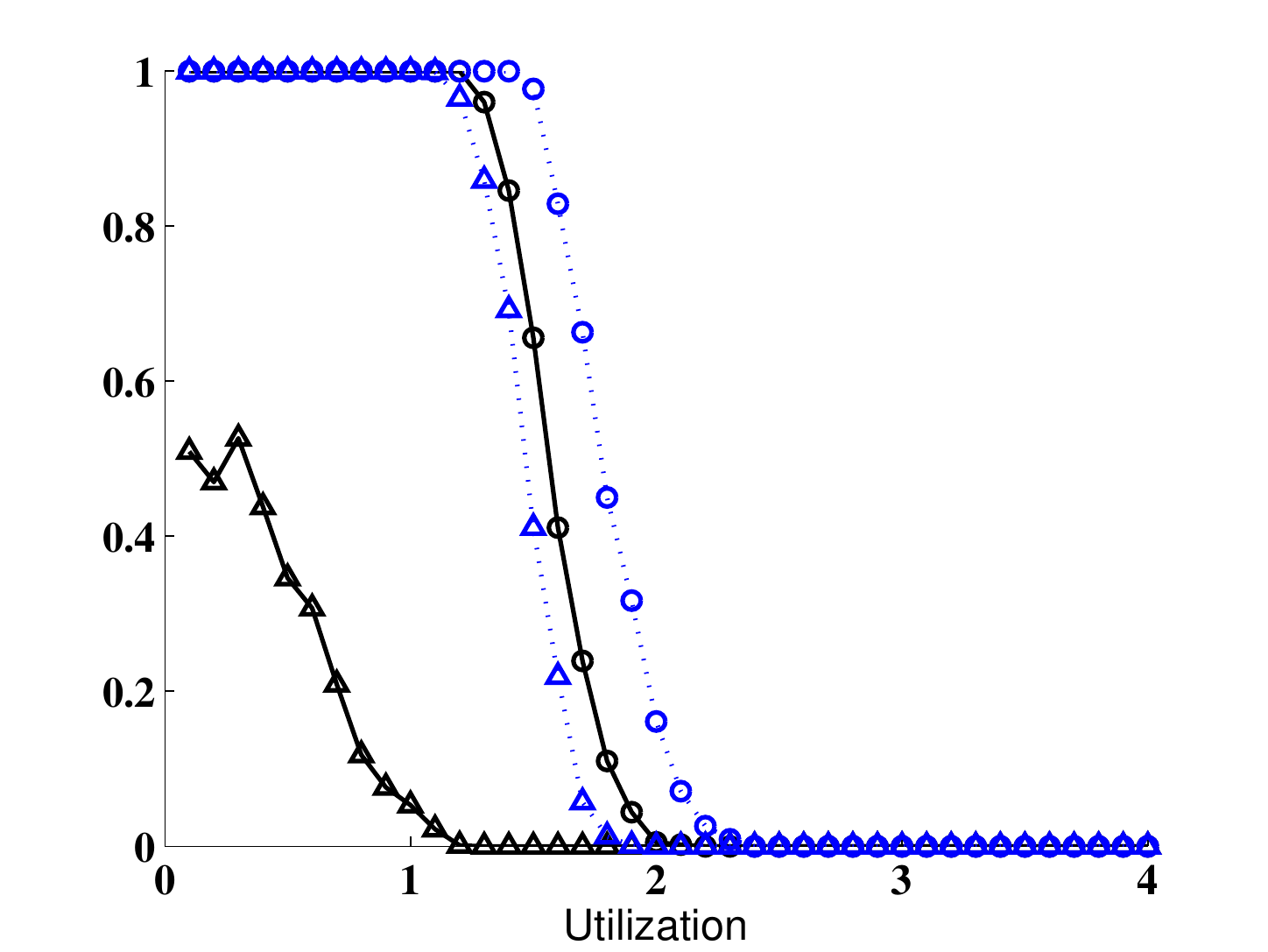}} \hspace{0mm}

  \caption{Schedulability result. In all nine graphs, the $y$-axis denotes the fraction of generate task systems that were schedulable under two schedulability tests in different conditions and $x$-axis denotes the utilization cap. In the first (respectively, second, third) row of graphs, the the value of $\alpha$ is 0.9 (respectively, 0.5, 0.2). In the first (respectively, second and third) column of graphs, light (respectively, medium and heavy) per-task   utilizations are assumed. Each graph gives two curves per tested approach for the cases of short and long suspensions, respectively.} 
  \vspace{1mm} \normalsize
  \label{fig:exp}
\end{figure*}
\paragraph*{\textbf{Experimental setup}}Consider that for magnetic disks, the suspension delay incurred is roughly 10$\mu s$ to 45$\mu s$\cite{kang:disktime}. Thus in our experiment, the lengths of writing phases $S_i$ were uniformly distributed over $[5\mu s,50\mu s]$.  Task utilizations $U_i$ were generated using three uniform
distributions: [0.001, 0.05](light), [0.05, 0.1](medium) and [0.1, 0.3](heavy). The suspension ratio $V_i$ of each task were generated by following two uniform distributions: $[0.005, 0.1]$(short) and $[0.1 , 0.3]$(long). And for each task $\tau_i$, we set the ratio of the first computing phase to its total computation length $\alpha_i$ ($\alpha_i=\frac{C_{i,1}}{C_i}$) as $0.9$, $0.5$ and $0.2$, respectively. Therefore, $T_i$ can be calculated by $S_i$ and $V_i$, and $C_i$ can be calculated by $T_i$ and $U_i$. The length of $C_{i,1}$ and $C_{i,2}$ can be calculated by $C_i$ and $\alpha_i$,  and $\delta_i$ can be calculated by $W_i$ and $C_{i,1}$. Note that when $\alpha_i$ grows larger, $\delta_i$ becomes smaller.

For each task system and a given utilization cap, we generated tasks to each task system until the total utilization exceeds the utilization cap and then we reduced the utilization of the last task to make the total utilization equal to utilization cap. For every utilization cap we generated 1000 task systems to evaluate the proposed schedulability test. We conducted the experiment for $m=4$.
\paragraph*{\textbf{Results}}The experimental results are shown in Fig. \ref{fig:exp} and the detailed explanation is available in the caption of Fig. \ref{fig:exp}. 

From Fig. \ref{fig:exp} we can see that in most cases our schedulability test is superior to suspension-oblivious density test, especially when $\alpha_i= 0.9$ and $\alpha_i= 0.5$. In Fig. \ref{fig:exp} (a), (b) and (c), regardless of the suspension length, our test significantly improves upon the suspension-oblivious density test in reducing the utilization loss. For example, in Fig. \ref{fig:exp}(a), when suspension is short, all task systems with $U_{sum}< 3.5$ are schedulable under our test; while the total utilization of task systems that are schedulable under suspension-oblivious density test is at most $1.9$. 

However, when the utilization is heavy and $\alpha_i$ is large, the suspension density-test becomes better than our schedulability test, as shown in Figs. \ref{fig:exp} (h) and (i). Moreover, with the increase of per-task utilization, both tests performed worst due to the large values of $Z_{max}$ in Theorem \ref{theorem:dtest} and  $U_{i}$ in Theorem \ref{theorem:wotest}. This is because with the decrease of $\alpha_i$, $\delta_i$ becomes larger. In such cases fewer task systems can pass our test due to the increase of the term $L$ defined in Lemma \ref{lem:proof1}. 

To conclude, in most cases our schedulability test is superior to the suspension-oblivious density test, often by a substantial margin. However, as discussed in Sec. \ref{subsec:rwanalysis}, these two tests do not dominate each other and in some extreme cases suspension-oblivious density test can be a better choice.
 
\section{Conclusion}
\label{sec:con}
In this paper, we have considered the problem of supporting applications with read/write operations in embedded real-time systems. First, we have shown that write-only task systems can be supported under GEDF on a multiprocessor with $O(m)$ suspension-related utilization loss. As demonstrated by experiments presented herein, in most cases our schedulability test is prior to the previous test. Second, in order to support read-write applications, we design a flexible I/O placement and a corresponding scheduling algorithm which enable us to completely eliminate the negative impact due to read- and write-induced suspensions. The presented case study shows our I/O placement is able to significantly reduce the response time. The presented case study implemented in real systems suggest that our proposed I/O-placement-based GEDF-R/W scheduling is feasible in practice. 

In this paper, we assume that the resource of I/O devices is sufficient to support as many tasks as we need, which is not true in practice. To handle I/O contention, one possible way is to integrate such contention into the worst-case suspension length parameter. However, this is very pessimistic. Thus, for future work, we plan to consider the co-scheduling problem on multiple resources. We also plan to design better algorithms that can reduce the job migration cost.

\bibliographystyle{abbrv}
\bibliography{sigproc}

\appendix
\section{The proof of Lemma 5}
\begin{proof}
To prove this lemma, we try to find a necessary condition for deadline miss by comparing the GEDF-R/W schedule to the PS schedule. Thus, we prove by contradiction. Assume job $\tau_{i,j}$ is the first job that misses its deadline at $d_{i,j}$. If more than one jobs misses deadline at $d_{i,j}$, we choose the one with the highest priority. Then we get rid of the jobs with priorities lower than that of $\tau_{i,j}$.

According to Lemma \ref{lem:pre_analysis}, $lag(\tau_{i,j}, d_{i,j}, S)> 0 $. And because $\tau_{i,j}$ is the first job misses its deadline, $lag(\tau_i, d_{i,j}, S)> 0 $ and for every $k\neq i$, $lag(\tau_k, d_{i,j}, S) = 0$. Thus, by Eq. (\ref{eq:LAG for task set_1}), $LAG(\tau, d_{i,j}, S) > 0$.

From time instant $0$, let $t^{*}$ be the earliest time instant that
\begin{equation}
\label{eq:as_lag2_2}
LAG(\tau, t^{*}, S)> 0.                   
\end{equation}
Since $LAG(\tau, 0,S) = 0$, $d_{i,j} \geq t^{*} > 0$. By the definition of $LAG(\tau, t^{*}, S)$, there exists a task $\tau_k$ at $t$ that $lag(\tau_k, t^{*}, S)>0$. Let $\tau_{k,l}$ be the pending job of $\tau_k$ at $t^{*}$. Because jobs in $\tau_k$ released before $\tau_{k,l}$ have finished all their computation and suspension, we have
\begin{eqnarray}
\label{eq:as_lag3_2}
lag(\tau_k, t^{*}, S)  = lag(\tau_{k,l}, t^{*}, S) > 0. 
\end{eqnarray}

There are three kinds of unit intervals to be considered in $[r_{k,l}, t^{*})$: (1)$\tau_{k,l}$ is not comp-preempted and $\tau_{k,l}$ suspends in it;(2)$\tau_{k,l}$ is not comp-preempted and $\tau_{k,l}$ computes in it;(3)$\tau_{k,l}$ is comp-preempted in it. Let $\beta_1, \beta_2, $and $\beta_3$ be the sets of each kind of unit intervals, respectively. Thus $\beta_1\cup \beta_2 \cup \beta_3= [r_{k,l}, t)$ and they are pairwise disjoint. Let $B_1$, $B_2$ and $B_3$ be the length of each set, respectively. Note that unit intervals in $\beta_3$ must be busy and $\tau_{k,l}$ may suspend in intervals in $\beta_3$.

\paragraph*{\textbf{Case A}}First we consider the case when $B_1 >0$. If $B_1 >0 $ , according to GEDF-R/W, $\tau_{k,l}$ must have finished all its computation by $t^{*}$. Thus,
\begin{eqnarray*}
&& lag(\tau_{k,l}, t^{*}, S) \\
&=& A(\tau_{k,l}, r_{k,l}, t^{*}, PS)- A(\tau_{k,l}, r_{k,l}, t^{*}, S) \\
&=& U_k \cdot B_2 - C_k \\
&\leq& U_k \cdot T_k - C_k \\
&=& 0
\end{eqnarray*} 
which violates Eq. (\ref{eq:as_lag3_2}).
\paragraph*{\textbf{Case B}}Then, we discuss the cases when $B_1= 0$. \ref{fig:lemma5}, which means that $\tau_{k,l}$ remains comp-pending during $[r_{k,l}, t)$. 

\textbf{case B.1: $B_1= 0,B_2 = B_3= 0$.} In this case, $t^{*}= r_{k,l}$ and $lag(\tau_{k,l}, t^{*}, S)= 0$ which violates Eq. (\ref{eq:as_lag3_2}). 

\textbf{case B.2: $B_1= 0, B_2 > 0, B_3= 0$.} In this case,$[r_{k,l}, t^{*})$ is a busy interval for $\tau_{k,l}$. So $lag(\tau_{k,l}, t^{*}, S) < 0$, which violates Eq. (\ref{eq:as_lag3_2}). 

\textbf{case B.3:} $B_1= 0 , B_2 \geq 0, B_3 > 0$ as shown in Fig. \ref{fig:lemma5}. In this case, first we consider $lag(\tau_{k,l}, t^{*}, S)$. By the definition of $A(\tau_{k,l}, t_1, t_2, PS)$ and $A(\tau_{k,l}, t_1, t_2, S)$,
\begin{eqnarray}
\label{eq:a_c4_1_2}
A(\tau_{k,l}, r_{k,l}, t^{*}, PS) = (B_2+B_3) \cdot u_k  
\end{eqnarray}
and $A(\tau_{k,l}, r_{k,l}, t^{*}, S)= B_2$ . By Eq. (\ref{eq:as_lag3_2}),
\begin{eqnarray*}
\label{eq:a_c4_2_2}
0 &<& lag(\tau_{k,l}, t^{*}, S)  \nonumber\\
&=& A(\tau_{k,l}, r_{k,l}, t^{*}, PS) - A(\tau_{k,l},  r_{k,l}, t^{*}, S)  \nonumber\\
&\leq& (B_2+B_3) \cdot U_k - B_2   \nonumber \\ 
&=& B_2 \cdot (U_k -1) + B_3 \cdot U_k 
\end{eqnarray*} 
which implies
\begin{equation}
\label{eq:a_c4_3_2}
B_2 < \frac{B_3 \cdot U_k}{1- U_k }.
\end{equation}
Now we consider $LAG(\tau, t^{*}, S)$. Because $t^{*}$ is the earliest time instant that $LAG(\tau, t^{*}, S)> 0$ and $r_{k,l} < t^{*}$, we have $LAG(\tau, r_{k,l}, S) \leq 0$. By the definition of LAG, $LAG(\tau, t^{*}, S) = LAG(\tau, r_{k,l}, S)+ A(\tau, r_{k,l}, t^{*}, PS) - A(\tau, r_{k,l}, t^{*}, S)$. Thus, 
\begin{eqnarray}
\label{eq:a_c4_4_2}
LAG(\tau, t, S) \leq A(\tau, r_{k,l}, t^{*}, PS) - A(\tau, r_{k,l}, t^{*}, S).
\end{eqnarray}
Also we have
\begin{eqnarray}
\label{eq:a_c4_5_2}
A(\tau, r_{k,l}, t^{*}, PS) = U_{sum} \cdot (B_2+B_3)
\end{eqnarray} 
and
\begin{eqnarray}
\label{eq:a_c4_6_2}
A(\tau, r_{k,l}, t^{*}, S)  \geq m \cdot B_3+ B_2.
\end{eqnarray}
Therefore, by Eq. (\ref{eq:as_lag2_2}),
\begin{eqnarray*}
\label{eq:a_c4_7_2}
0 &<& LAG(\tau, t^{*}, S) \nonumber\\
&&{\lbrace \rm{by}~(\ref{eq:a_c4_4_2}) \rbrace} \nonumber\\ 
&\leq& A(\tau, r_{k,l}, t^{*}, PS) - A(\tau, r_{k,l}, t^{*}, S)  \nonumber\\
&&{\lbrace \rm{by}~(\ref{eq:a_c4_5_2}) \rbrace} \nonumber\\ 
&=& U_{sum} \cdot (B_2+B_3)- A(\tau, r_{k,l}, t^{*}, S) \nonumber \\ 
&&{\lbrace \rm{by}~(\ref{eq:a_c4_6_2}) \rbrace} \nonumber\\
&\leq& U_{sum} \cdot (B_2+B_3)- (m \cdot B_3+ B_2) \nonumber \\  
&=& (U_{sum}-1)\cdot B_2+ (U_{sum}-m) \cdot B_3\nonumber \\
&&{\lbrace \rm{by}~(\ref{eq:a_c4_3_2}) \rbrace} \nonumber\\ 
&<& \frac{(U_{sum}-1) \cdot B_3 \cdot U_k}{1- U_k }  + (U_{sum}-m) \cdot B_3
\end{eqnarray*}
which implies,
\begin{eqnarray*}
\label{eq:a_c4_8_2}
&&U_{sum} > m-(m-1)\cdot U_k > m-(m-1)\cdot U_{max}
\end{eqnarray*}
However, this violates Eq. (\ref{eq:as_sum_2})

From the above, we have discussed all the possible cases and each case implies a contradiction. Therefore, Lemma \ref{proof2} follows. 
\end{proof}

\begin{figure}[th]
	\begin{center}
	\includegraphics[width=3.5in]{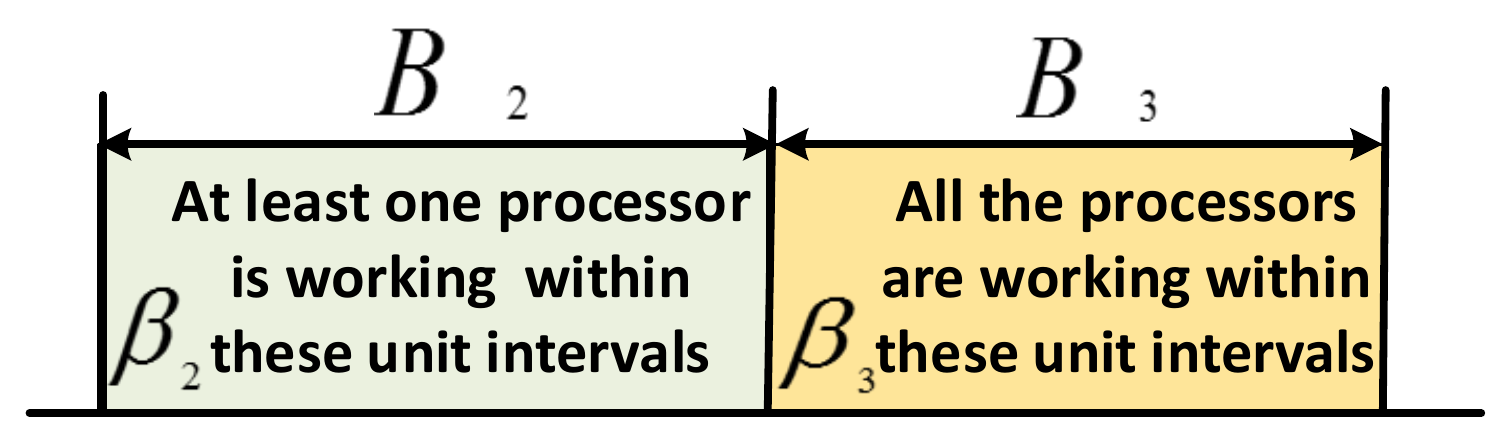} 
	\end{center} 
\vspace{-5mm}
\caption{Two sets of unit intervals in $[r_{k,l}, t^{*})$}
\vspace{-3mm}
\label{fig:lemma5}
\end{figure}

\end{document}